\documentclass[11pt]{article}

\usepackage[utf8]{inputenc}
\usepackage[T1]{fontenc}
\usepackage[margin=1in]{geometry}
\usepackage{amsmath,amssymb,amsthm}

\usepackage[numbers,sort&compress]{natbib}
\usepackage[colorlinks,citecolor=blue,urlcolor=blue,linkcolor=blue]{hyperref}
\usepackage{graphicx}

\usepackage{thm-restate}

\theoremstyle{plain}

\newtheorem{theorem}{Theorem}[section]
\newtheorem{lemma}[theorem]{Lemma}
\theoremstyle{definition}
\newtheorem{definition}[theorem]{Definition}

\theoremstyle{remark}

\newcommand{\adj}{A}
\newcommand{\obs}{X}
\newcommand{\probs}{\theta}
\newcommand{\probshat}{\widehat{\probs}}
\newcommand{\probstilde}{\widetilde{\probs}}

\newcommand{\indicator}[1]{\mathbbm{1}_{\left\{#1\right\}}}
\newcommand{\abs}[1]{\left|#1\right|}

\newcommand{\map}{z}
\newcommand{\maps}{\mathcal{Z}}

\newcommand{\indexk}[1]{^{(#1)}}

\newcommand{\gbern}[1]{\operatorname{Multinoulli}\left(#1\right)}
\newcommand{\N}{\mathbb{N}}
\newcommand{\n}{[n]}
\newcommand{\measures}[1]{\mathcal{P}\left(#1\right)}

\newcommand{\K}{\mathcal{K}}
\newcommand{\numK}{L}
\newcommand{\indexK}{l}

\newcommand{\setobs}{\left\{ a \in \{0,1\}^\numK: \|a\|_1 = 1\right\}}

\newcommand*\diff{\mathop{}\!\mathrm{d}}

\NewDocumentCommand{\boundedssm}{ O{d} O{k}}{\mathcal{B}_{#1}(SSM(#2))}
\newcommand{\ubshapes}[2]{\operatorname{UB}(#1,#2)}
\newcommand{\lbshapes}[2]{\operatorname{LB}(#1,#2)}

\newcommand{\E}[2][]{\mathbb{E}_{#1}\left[#2\right]}
\newcommand{\pr}[1]{\mathrm{pr}\left(#1\right)}

\newcommand{\cdf}[2][]{F_{#1}\left(#2\right)}
\newcommand{\Wd}[2][]{\mathcal{W}_{#1}\left(#2\right)}
\newcommand{\histmeasures}[2][\numK]{\mathcal{P}_{#1}\left(#2\right)}

\usepackage{tikz}
\usetikzlibrary{arrows, fit, matrix, positioning, shapes, backgrounds, positioning, shapes.geometric, arrows.meta, calc}
\definecolor{IBM1}{RGB}{100, 143, 255}
\definecolor{IBM2}{RGB}{220, 38, 127}
\definecolor{IBM3}{RGB}{255, 176, 0}

\usepackage[nameinlink]{cleveref}
\crefname{assumption}{assumption}{assumptions}

\usepackage{xfp}

\newcommand{\rom}[1]{\uppercase\expandafter{\romannumeral #1\relax}}
\newtheorem{assumption}{Assumption}
\newtheorem{proposition}{Proposition}
\newtheorem{remark}{Remark}
\usepackage{algorithmicx}
\usepackage[ruled]{algorithm2e}
\usepackage{bbm}
\crefname{algorithm}{Algorithm}{Algorithms}

\begin{document}

\title{Graphon estimation beyond binary edges: inference for decorated graphs with applications to multiplex and weighted networks}

\author{%
  Charles Dufour\thanks{Institute of Mathematics, \'Ecole Polytechnique F\'ed\'erale de Lausanne. Email: \texttt{charles.dufour@epfl.ch}. ORCID: 0009-0000-3612-7335.}
  \and
  Sofia C. Olhede\thanks{Institute of Mathematics, \'Ecole Polytechnique F\'ed\'erale de Lausanne. Email: \texttt{sofia.olhede@epfl.ch}. ORCID: 0000-0003-0061-227X.}%
}

\date{}

\maketitle

\begin{abstract}
We introduce the first doubly non-parametric estimation method for decorated graphons, a generalisation of graphons that encodes edge weights, edge types, and other edge-level attributes in large networks. Graphons describe the limiting behaviour of large unlabelled networks through a symmetric measurable function governing the probability of edge formation, but the standard framework is restricted to binary edge information. Decorated graphons lift this restriction, yet no inference procedure has previously been available for them. The proposed estimator extends classical graphon estimation techniques to this enriched setting. We derive rates of convergence and show that, for compactly supported decorations, these rates agree with known non-parametric rates for estimating real-valued functions. Monte Carlo experiments confirm that the theoretical rates are attained in finite samples, and applications to synthetic and empirical networks show improved fit relative to binary-edge baselines. The method extends graphon-based inference to multiplex networks and attributed graphs simultaneously.
\end{abstract}

\medskip
\noindent\textbf{Keywords:} Decorated graph, graphon, nonparametric estimation, edge attributed, edge labels, multiplex network, weighted network.

\bigskip

\section{Introduction}

Graphons have emerged as a fundamental tool in studying large, unlabeled simple networks, offering a robust framework grounded in the theory of exchangeability \citep{kallenberg_representation_1989,diaconis_graph_2007,borgs_lp_sparse_I_2019}. These mathematical objects serve as the limiting objects for sequences of dense graphs which are instrumental in domains such as sociology, biology, and computer science, where understanding the asymptotic properties of large networks is crucial \citep{lovasz_large_2012}. Traditional graphons encode binary information on edges, indicating the presence or absence of connections between nodes. However, many real-world networks exhibit richer structures, where edges carry additional information beyond mere connectivity, such as weights or types. For example, in a social network, edges between individuals might not only indicate friendship (a binary state) but also the frequency of interaction (a weight) or the type of relationship (colleague, friend, co-authorship) \citep{resnick_protecting_1997,magnani_combinatorial_2013}. In biological networks, edges might represent different biochemical interactions (e.g., protein-protein interactions, gene regulation), with decorations capturing the interaction strength or type. In transportation networks, edges can be decorated with travel time, cost, or type of transportation mode (e.g., bus, train, flight) \citep{cardillo_emergence_2013}. 
 
To address these complexities, Lov\'asz and Szegedy \citep{lovasz_limits_2010} introduced the theoretical framework of decorated graphons (see also probability-graphons \citep{abraham_probabilitygraphons_2025}), which extend traditional graphons by allowing edges to carry more detailed information (so-called decorations, features, weights, etc). Multiplex networks naturally fit into this framework \citep{kivela_multilayer_2014,ganguly_multiplexons_2025};  in multiplex networks, multiple types of connections can exist between the same set of nodes. By enumerating all possible combinations of these various connections and representing them as decorated edges, we can effectively treat multiplex networks as decorated graphs.

Several lines of work extend graph models beyond binary edges. For graphs with edge attributes, Donier-Meroz et al.~\citep{donier-meroz_graphon_2023} estimate graphon-type objects but target only the conditional mean of the edge variable, not its full distribution. Xu et al.~\citep{xu_optimal_2020} derive optimal rates for community detection in weighted stochastic block models, again without estimating the graphon. For multilayer and multiplex networks, Barbillon et al.~\citep{barbillon_stochastic_2017} propose a multiplex SBM with a fixed number of communities but do not consider the nonparametric limit. Other contributions in the multilayer setting include \citep{avrachenkov_community_2022,chandna_edge_2022,skeja_quantifying_2024,wang_multilayer_2024}. Pensky~\citep{pensky_signed_2025} and Fishkind et al.~\citep{fishkind_complete_2021} study specific subclasses of decorated graphs (signed and correlated networks, respectively). Lubberts et al.~\citep{lubberts_random_2025} approach edge-attributed graphs through random line graphs. Time-varying network models have also been developed \citep{chandna_nonparametric_2020,pensky_dynamic_2019,suveges_networks_2023}. None of these methods provide a unified nonparametric estimator for the full decorated graphon with convergence guarantees.

\subsection{Summary of contributions}
To the best of our knowledge, this work introduces the first estimation method specifically designed for decorated graphons with compactly supported decorations. While previous approaches to weighted graphs have focused primarily on estimating the mean of the adjacency matrix \citep{donier-meroz_graphon_2023}, or have been restricted to a fixed number of communities without considering the graphon limit \citep{barbillon_stochastic_2017}, our methodology characterizes the entire distribution of edge attributes. By extending the network histogram \citep{olhede_network_2014} to decorated edges, we preserve the convergence rates of traditional graphon estimation \citep{gao_rateoptimal_2015,klopp_oracle_2017,verdeyme_hybrid_2024}. Our main contributions are as follows:
\begin{itemize}
	\item Nonparametric Estimator for Decorated Graphons: We introduce the first nonparametric estimator for decorated graphons, recovering the full conditional distribution of edge attributes rather than only their mean \citep{donier-meroz_graphon_2023} or a parametric block structure with a fixed number of communities \citep{barbillon_stochastic_2017}. The estimator is based on least-squares minimization over piecewise-constant approximations of the graphon, embedding each decoration as a vertex of the probability simplex.
	\item Theoretical Guarantees for the Finite Setting: We establish finite-sample concentration bounds for our least-squares estimator. The error decomposes into a clustering rate (which matches the binary graph setting) and a nonparametric rate that scales linearly with the size of the decoration space. Because the one-hot observations are dependent within each edge, the concentration arguments of \citep{gao_rateoptimal_2015,verdeyme_hybrid_2024} based on Hoeffding's inequality are replaced throughout by McDiarmid's inequality.
	\item Extension to Compactly Decorated Graphons: We provide a rigorous methodology for networks with continuous edge variables (e.g., weights) by introducing a non-parametric approximation approach. This approach approximates continuous distributions using histograms while explicitly modeling the inherent sparsity of real-world networks via an atom at zero.
	\item Convergence Rates via Wasserstein Metric: For the continuous setting, we derive convergence rates utilizing the Wasserstein-1 distance, which naturally accommodates distributions with atomic components without requiring the existence of a probability density.
	\item Application to Multiplex Networks: We validate our theoretical findings on both synthetic graphs and a real-world multiplex network of human diseases, demonstrating that our estimator can capture cross-layer dependencies to successfully predict missing links.
\end{itemize}

The paper is organized as follows. Section~\ref{section:background} introduces decorated graphs and graphons. Section~\ref{section:inference_finitely} develops our estimator for finitely decorated graphons with convergence guarantees. Section~\ref{section:inference_compact} extends the framework to compactly supported decorations via discretization. Section~\ref{section:experiments} presents simulations and Section~\ref{section:ada} a real-data application. Proofs and practical considerations are deferred to the Appendices.

\section{Background: Decorated Graphs and Graphons}
\label{section:background}

A standard simple graph on $n$ vertices is represented by a symmetric adjacency matrix $\adj \in \{0,1\}^{n \times n}$. In a \emph{decorated graph} \citep{lovasz_limits_2010} (also called an edge-attributed graph), each edge $(i,j)$ instead carries an element of a set $\K$ (a weight, a color, or a type label, for instance), so that the adjacency matrix takes values in $\K^{n\times n}$. The set $\K$ contains a distinguished element $0_{\K}$ representing the absence of an edge. Requiring the decorated graph to be exchangeable means that $\{\adj_{ij}\}$ and $\{\adj_{\pi(i)\pi(j)}\}$ have the same joint distribution for every permutation $\pi$. We now describe the generating mechanism of such $\K$-valued exchangeable arrays.

\begin{definition}[$\K$-graphon \citep{lovasz_limits_2010}]
	\label{def:decorated_graphon}
	Let $\measures{\K}$ denote the set of probability Borel measures on a compact space $\K$, and let $\mathcal{W}(\K)$ denote the set of two-variable Borel measurable functions (with the weak topology)
	\begin{equation*}
		W:[0,1]^2 \rightarrow\measures{\K}, \text{ such that } W(x, y)=W(y, x) \text{ for every } (x, y) \in[0,1]^2.
	\end{equation*}
	Elements of $\mathcal{W}(\K)$ are referred to as \emph{decorated graphons} (or \emph{$\K$-graphons} when the decoration space needs emphasis).
\end{definition}

\begin{remark}
    \label{remark:decorated_network_complete_or_not}
    One can equivalently model edge existence and decoration separately, or simply let the decoration $0_{\K}$ encode the absence of an edge. We adopt the latter convention throughout: decorations and edges are used interchangeably.
\end{remark}

As for the simple graph case, we can get a functional representation of a $\K$-decorated graph; from Kunszenti-Kov\'acs, \citep{kunszenti-kovacs_multigraph_2022} and Kallenberg \citep{kallenberg_probabilistic_2005}, we get that if $\K$ is a compact Hausdorff space and $\adj \in \K^{n \times n}$ is jointly exchangeable, there exists a $\K$-graphon $W$ such that
\begin{equation}
	\label{eq:aldous-hoover-like}
	\adj_{ij} \mid \xi_i, \xi_j \overset{\text{iid}}{\sim} W(\xi_i,\xi_j), \text{ where } \xi_i \overset{\text{iid}}{\sim} U[0,1],
\end{equation}
where $\adj_{ij} \sim W(x,y)$  indicates that $\adj_{ij}$ follows the probability distribution characterized by $W(x,y)$. When observing a realization from a decorated graphon, each edge will then be decorated with exactly one element of $\K$.  The representation in eq.~\eqref{eq:aldous-hoover-like} resembles the Aldous-Hoover theorem in the case of simple graphs \citep{aldous_representations_1981,hoover_relations_1979}. Indeed, if $\K = \{0,1\}$, $\adj$ is the simple adjacency matrix, and $W(\xi_i,\xi_j)$ is a Bernoulli distribution.

\subsection{Finitely Decorated Graphs}
\label{subsection:finitely_decorated_graphs}

Let $\K$ be a finite set, such that $\K = \{x_1, \ldots, x_\numK\}$ with cardinality $|\K| = \numK < \infty$. Although this may appear restrictive, finite decorations naturally encompass discretized weighted graphs, colored graphs, and multiplex networks \citep{bianconi_multilayer_2018}: in the latter case, each unique combination of edge types across layers corresponds to a distinct decoration.

\begin{figure}[h!]
    \centering
    \includegraphics[width=\linewidth]{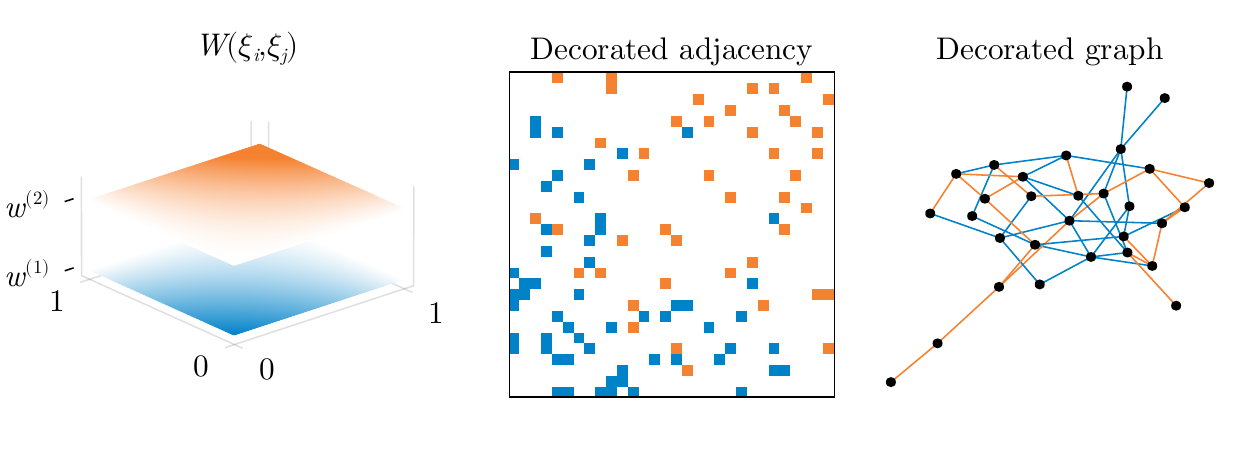}
    \caption{Example of a finitely decorated graphon (left) and a realized decorated graph (centre and right) from this graphon according to eq.~\eqref{eq:aldous-hoover-like}. Here $\K = \{\text{white}, \text{blue}, \text{orange}\}$, and each edge is colored with exactly one color. The color $\text{white}$ represents the absence of edge ($0_{\K}$), and thus $w^{(0)}$ is not shown.}
    \label{fig:decorated_example}
\end{figure}

To establish a probability distribution over $\K$, it is sufficient to define $\numK$ probabilities. Accordingly, a $\K$-graphon $W$ is characterized by $\numK$ symmetric measurable functions $w\indexk{\indexK} : [0,1]^2 \rightarrow [0,1]$, where the sum over all functions at any pair $(x, y)$ equals 1, i.e., $\sum_{\indexK=1}^\numK w\indexk{\indexK}(x, y) = 1$; here $w\indexk{\indexK}$ represents the probability of realization of the $\indexK$th decoration $x_{\indexK}$ (see Figure~\ref{fig:decorated_example}). Given fixed $x, y \in [0,1]$, $W(x, y)$ corresponds to an element of the set $\Delta_{\numK} := \{p \in [0,1]^\numK : \|p\|_1 = 1\}$ \citep{lovasz_limits_2010}. When conditioned on the latent variables $\{\xi_i\}$, we define:
\begin{equation*}
	\probs_{ij} = W(\xi_i, \xi_j) \text{ for } i, j \in \n,
\end{equation*}
where $\probs_{ij} \in \Delta_{\numK} \subset [0,1]^\numK$. 

\subsection{Estimation Objectives: Matrix vs. Function}
\label{subsection:estimation_goals}

There are two related but distinct inference goals. The first is \emph{value estimation}: recovering the matrix $\probs$ where $\probs_{ij} = W(\xi_i, \xi_j)$, a denoising problem measured by the Mean Squared Error (MSE):
\begin{equation}
  \label{equation:MSE_def}
  \frac{1}{n^2}\sum_{i,j \in \n}\|\probshat_{ij}-\probs_{ij}\|_2^2 = \frac{1}{n^2}\|\probshat-\probs\|_F^2.
\end{equation}
For dense binary graphs, minimizing this MSE is asymptotically equivalent to maximum likelihood \citep{gaucher_optimality_2021}. The second is \emph{function estimation}: recovering $W$ over $[0,1]^2$ up to the inherent unidentifiability from relabeling nodes. Since any measure-preserving bijection $\phi$ yields an equivalent graphon $W^\phi(x,y) = W(\phi(x), \phi(y))$, performance is measured by the Mean Integrated Square Error (MISE):
\begin{equation}
    \label{equation:MISE_def}
    \text{MISE}(\hat{W}, W) = \
  \inf_{\phi \in \mathcal{M}}\iint_{(0,1)^2}\left\|W^{\phi}\left(x,y\right)-\hat{W}\left(x,y\right)\right\|_2^2\diff{x}\diff{y},
\end{equation}
where $\mathcal{M}$ is the set of measure-preserving bijections from $[0,1]$ to $[0,1]$ \citep{wolfe_nonparametric_2013,klopp_oracle_2017}. Our method addresses both goals simultaneously: we construct a matrix estimator $\probshat$ via a piecewise-constant approximation of the graphon, which naturally induces a functional estimator $\hat{W}$.

\section{Inference of Finitely Decorated Graphons}
\label{section:inference_finitely}

In this section, we introduce the methodology for estimating finitely decorated graphons. Our goal is to recover the underlying probability function $W$ given a single observation of a decorated graph adjacency matrix $\adj \in \K^{n \times n}$. We formalize the assumption on the decoration space $\K$ as follows:

\begin{assumption}
	\label{assumption:finite}
	The decoration space $\K$ is a finite set, i.e.,  $|\K|=\numK < \infty$, $\K = \{x_1,\ldots,x_\numK\}$.
\end{assumption}

\subsection{Approximation via piece-wise constant functions}
To make estimation feasible, we approximate the smooth graphon $W$ using piece-wise constant functions. The standard approach is a stochastic block model (SBM) that partitions $[0,1]^2$ into a regular grid of $k\times k$ square blocks \citep{wolfe_nonparametric_2013,gao_rateoptimal_2015,janson_can_2021}, assigning an independent parameter to each block. In the decorated setting, the block parameter is a probability vector in the simplex $\Delta_{\numK}$ rather than a scalar, but the geometry of the approximation is unchanged: a $\K$-decorated $k$-block SBM is a graphon $W:[0,1]^2 \to \measures{\K}$ of the form $W(x,y) = \theta_{u_k(x), u_k(y)}$ with $u_k(x) = \lceil kx \rceil$ and symmetric parameters $\{\theta_{ab}\}_{a,b\in[k]} \subset \Delta_{\numK}$.

A more flexible approximation, the \emph{Stochastic Shape Model} (SSM) of Verdeyme and Olhede \citep{verdeyme_hybrid_2024}, allows multiple disjoint blocks of a $k\times k$ grid to share a common parameter, grouping them into a single ``shape'' (see Figure~\ref{fig:ssm}). This pooling reduces the effective number of parameters when the graphon has symmetries or large homogeneous regions and improves the bias-variance tradeoff. Our estimator and theory are stated in this more general form since the SBM is recovered as a special case, but no aspect of the contribution depends on the additional flexibility of shapes.

\begin{definition}[$\K$-decorated $(s,k)$-Stochastic Shape Model]
	\label{definition:ssm}
	Let $u_k: [0,1] \to [k]$ be the block assignment $u_k(x) = \lceil kx \rceil$, and let  $u_s: [k]^2 \to [s]$ be a symmetric mapping, i.e., $u_s(a,b) = u_s(b,a)$ for all $a,b \in [k]$. A graphon $W: [0,1]^2 \to \measures{\K}$ is a $\K$-decorated $(s,k)$-stochastic shape model if there exist measures $\{\theta_c\}_{c=1}^s \subset \measures{\K}$  such that
    \[
        W(x, y) = \theta_{u_s(u_k(x),\, u_k(y))}
        \quad \text{for all } (x,y) \in [0,1]^2.
    \]
    When $s = \binom{k+1}{2}$, the SSM reduces to a standard $k$-block SBM; see Appendix~\ref{sec:finitely_decorated} for the mapping from nodes to shapes.
\end{definition}

\begin{figure}[h!]
	\centering
	\includegraphics[width=14cm]{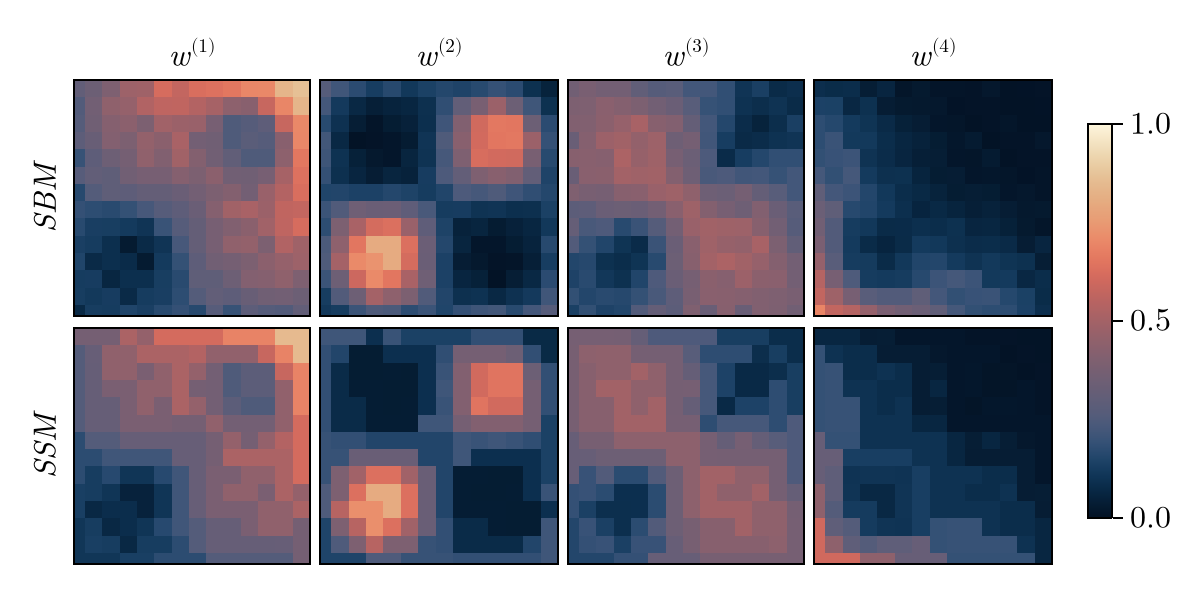}
	\caption{SBM approximation of $W_3$ (see Table~\ref{tab:sim}) in the first row and estimated decorated graphon using a decorated $(s=27,k=14)$-SSM in the second. Our theory applies to both; the SSM refinement offers a better bias-variance tradeoff when the graphon has symmetries, as visible in $w\indexk{4}$. The estimator was computed based on an observation with $300$ nodes.}
	\label{fig:ssm}
\end{figure}

\subsection{Construction of the Estimator}
\label{subsection:estimator_construction}

Our estimation strategy relies on a least-squares approach, which for dense binary networks is equivalent to maximum likelihood estimation \citep{gaucher_optimality_2021}. However, a direct comparison between observed decorations $\adj_{ij} \in \K$ and probabilities $W(\xi_i, \xi_j) \in [0,1]^\numK$ is not feasible, as the former lies in a discrete set while the latter lies in the probability simplex $\Delta_\numK$.

To bridge this gap, we employ a \emph{one-hot encoding} transformation. We identify each element of $\K$ with the corresponding vertex of $\Delta_\numK$ via the embedding $\iota:\K \hookrightarrow \Delta_\numK$, $\iota(x_l) = e_l$, where $e_l\in \{0,1\}^{\numK}$ denotes the $l$-th standard basis vector of $\mathbb{R}^\numK$. Writing $\obs_{ij} = \iota(\adj_{ij})$, the transformed observations lie at the vertices of the same simplex on which $W$ takes values, and $\obs_{ij} \mid \xi_i, \xi_j \sim \text{Multinomial}(1, W(\xi_i, \xi_j))$. We write $\obs = (\obs_{ij})_{i,j \in [n]}$ for the resulting array.

We perform the estimation by minimizing the least-squares error over the set of matrices generated by Stochastic Shape Models. Let $\Theta_{s,k}$ denote the set of all possible $n \times n$ probability matrices that can be generated by a $(s,k)$-SSM:
$$
\Theta_{s,k} = \left\{ \probs \in [0,1]^{n \times n \times \numK} : \exists\; \text{a } (s,k)\text{-SSM } W \text{ s.t. } \probs_{ij} = W(\xi_i,\xi_j) \text{ for some } \xi \in [0,1]^n \right\}.
$$

We define the matrix estimator $\probshat$ as the minimizer of the squared Frobenius error over this set:
\begin{equation}
	\label{eq:least-squares-formulation}
	\probshat \in \underset{{\probs \in \Theta_{s,k}}}{\operatorname{argmin}} \frac{1}{n^2}\|\obs-\probs\|_F^2.
\end{equation}

Finding this minimizer involves simultaneously identifying the optimal node-to-block assignments, the optimal block-to-shape mapping, and the optimal shape parameters $\{\theta_c\}$. For a fixed shape structure, the optimal parameters are simply the empirical means of the one-hot vectors within each shape. Thus, the problem reduces to a combinatorial optimization over partitions, which can be solved using clustering algorithms adapted from the SBM literature. We are effectively fitting a $(s,k)$-SSM to the observed data, and $\probshat$ is a ``projection'' of the parameters of this fitted model onto the space of $n \times n$ matrices, in such a way that the resulting matrix is aligned with the observed adjacency matrix $\adj$.

In practice, the combinatorial optimization in eq.~\eqref{eq:least-squares-formulation} is NP-hard, and we approximate the global optimum using a greedy label-switching algorithm \citep{bickel_nonparametric_2009,olhede_network_2014} initialized with a multilevel $k$-way partitioning via METIS \citep{karypis_metis_1997,karypis_multilevel_1998}; see Appendix~\ref{sec:optimization} for details. When the smoothness $\alpha$ is unknown, we select $k$ using the automatic bandwidth selection of Olhede and Wolfe~\citep{olhede_network_2014} applied to each decoration probability $w^{(l)}$, and choose $s$ via the Bayesian Information Criterion as in Verdeyme and Olhede~\citep{verdeyme_hybrid_2024}.

Once the matrix parameter $\probshat$ is obtained, we recover the graphon estimator $\hat{W}: [0,1]^2 \to [0,1]^\numK$ as the step function induced by the estimated block probabilities:

\begin{equation}
	\label{eq:function-estimator}
	\hat{W}(x,y) = \probshat_{\lceil nx \rceil,\lceil ny \rceil},
\end{equation}
which is equivalent to the $(s,k)$-SSM associated to $\probshat$ in eq.~\eqref{eq:least-squares-formulation}.

\subsection{Properties of the Estimator}
\label{subsection:finitely-properties}

We now analyze the theoretical properties of the least squares estimator defined above, under settings similar to the current literature \citep{olhede_network_2014, gao_rateoptimal_2015, klopp_oracle_2017, verdeyme_hybrid_2024}. We start by establishing a bound for the case where the true graphon is exactly a Stochastic Shape Model.

\begin{theorem}
	\label{theorem:shape-rate}
	Let $W$ be a $\K$-decorated $(s,k)$-stochastic shape model, with $\K$ satisfying Assumption~\ref{assumption:finite}. Then for any $C'>0$ there exists $C>0$ such that
	\begin{equation*}
		\frac{1}{n^2}\sum_{i,j \in \n}\left\|\probshat_{ij} - \probs_{ij}\right\|_2^2 \leq C\left(\frac{s\numK}{n^2} + \frac{\log(\max(k,s))}{n}\right),
	\end{equation*}
	with probability at least $1-\exp\left(-C'n\log (\max (k, s))\right)$. Furthermore,
	\[\sup _{\theta \in \Theta_{s,k}} \E{\frac{1}{n^2} \sum_{i j}\|\probshat_{i j}-\probs_{i j}\|_2^2} \leq C_1\left(\frac{s\numK}{n^2}+\frac{\log (\max (k, s))}{n}\right)\]
	with some universal constant $C_1>0$ and $n>\max \left(0, k^2-s\right)$.
\end{theorem}

The error bound in Theorem~\ref{theorem:shape-rate} reveals an interesting decomposition of the problem. It comprises two distinct terms: a nonparametric $s\numK/n^2$ and a clustering (\(\log(\max(k,s))/n\)).
Notably, the clustering rate is identical to that found in the simple binary graph case, both for standard block models \citep{gao_rateoptimal_2015} and extended shape models \citep{verdeyme_hybrid_2024}. This is intuitive: the difficulty of the clustering (finding the optimal partition of $n$ nodes into $k$ groups) remains fundamentally the same regardless of what information is on the edges. Conversely, the nonparametric rate scales linearly with $\numK$. This reflects the change in our modeling assumptions: for every shape, we must now estimate a probability vector of size $\numK$ rather than a single scalar parameter.

Next, we generalize this result to the setting where $W$ is a smooth function rather than a step function. We assume $W$ satisfies a Hölder continuity condition.

\begin{assumption}
	\label{assumption:holder}
	The $\K$-decorated graphon $W$ is Hölder continuous with exponent $\alpha \in (0,1]$. That is, $W \in \mathcal{H}(\alpha, M)$ where:
	$$ \sup _{(x, y) \neq\left(x^{\prime}, y^{\prime}\right) \in(0,1)^2} \frac{\|W(x, y)-W\left(x^{\prime}, y^{\prime}\right)\|_1}{\|(x, y)-\left(x^{\prime}, y^{\prime}\right)\|_1^\alpha} \leq M < \infty.$$
\end{assumption}

In the context of approximation, the complexity of a $(s,k)$-SSM is controlled by the diameter of its shapes (the maximum distance between any pair of points within a shape). Together with the number of shapes $s$ and blocks $k$, this diameter dictates the bias-variance tradeoff: finer shapes reduce bias (better approximation of smooth functions) but increase variance (fewer observations per parameter).

\begin{restatable}{theorem}{theoremHolderRateStatic}
	\label{theorem:holder-rate}
	For $W \in \mathcal{H}(\alpha, M)$ (Assumption~\ref{assumption:holder}) and $n> \numK$, for any $C^{\prime}>0$ there exists a constant $C>0$ depending only on $C^{\prime}, M, \alpha$, such that for an appropriate choice of $k$ and $s$ (specified in Appendix~\ref{appendix_a:oracle_ssm}), with all shapes having diameter at most $n^{-1/(\alpha \wedge 1 + 1)}$,
	$$
		\frac{1}{n^2} \sum_{i j}\|\probshat_{i j}-\probs_{i j}\|_2^2 \leq C\left(\numK n^{-2 \alpha /(\alpha+1)}+\frac{\log (n)}{n} \right),
	$$
    and
    $$
		\operatorname{MISE}\left(\widehat{W}_{\probshat}, W\right) \leq C\left(\numK n^{-2 \alpha /(\alpha+1)}+\frac{\log (n)}{n}  + n^{-\alpha \wedge 1}\right),
	$$
    with probability at least $1-\exp(-C'n)$. Furthermore,
	$$
		\sup _{W \in \mathcal{H}(\alpha, M)} \E{\frac{1}{n^2} \sum_{i, j \in}\|\probshat_{i j}-\probs_{i j}\|_2^2} \leq C_1\left(\numK n^{-2 \alpha /(\alpha+1)}+\frac{\log (n)}{n}\right),
	$$
    and
    $$
		\sup _{W \in \mathcal{H}(\alpha, M)} \E{\operatorname{MISE}\left(\widehat{W}_{\probshat}, W\right)} \leq C_1\left(\numK n^{-2 \alpha /(\alpha+1)}+\frac{\log (n)}{n} + n ^{-\alpha\wedge 1}\right).
	$$
\end{restatable}

The proof (given in Appendix~\ref{subsection:holder-rate}) decomposes the error into three terms: (i) a bias term from approximating the smooth graphon $W$ with a piecewise-constant function, controlled by $\alpha$ and the shape diameters; (ii) a variance term from estimating the parameters within each shape, scaling as $\numK s / n^2$; and (iii) a clustering term scaling as $\log(n)/n$. The optimal choice of the number of blocks $k$, the number of shapes $s$, and the shape diameters balances these three contributions. In the appendix, we introduce a parameter $\Delta\in[0,1]$ that interpolates between the standard block-model regime ($\Delta=0$, all blocks are separate shapes) and regimes with smaller blocks grouped into more complex shapes ($\Delta>0$); see Remark~\ref{remark:appendix_static_shape_number} and Appendix~\ref{appendix_a:oracle_ssm} for details.

When $\K = \{0,1\}$, and we pick $s$ such that we only consider SBMs,  these rates align perfectly with those established by Gao et al. \citep{gao_rateoptimal_2015} and by Klopp et al. \citep{klopp_oracle_2017}. Otherwise, when $s$ is allowed to vary, we retrieve the results of Verdeyme and Olhede \citep{verdeyme_hybrid_2024}, confirming that our method is a consistent generalization of the binary framework.

\section{Inference of Compactly Decorated Graphs}
\label{section:inference_compact}

We now extend our framework to the case where $\K$ is a compact Borel space. Without loss of generality, we consider $\K=[0,1]$ \citep{kallenberg_foundations_2021}. This setting is essential for analyzing weighted graphs where edges are continuous variables (e.g., travel times or interaction strengths).

\subsection{Estimation Strategy: Discretization}
Our strategy relies on a discretization approach. Since estimating a continuous distribution for every point on the graphon is ill-posed without further constraints, we approximate the continuous distributions using histograms. This transforms the compactly decorated problem into a finitely decorated one (by treating bins as categorical labels), allowing us to leverage the estimator $\probshat$ developed in Section~\ref{section:inference_finitely} (see Figure~\ref{fig:equivalence_compact_finite}).

\begin{figure}[h]
    \centering
    \begin{tikzpicture}
        \node[rectangle] (cont_dist) {\includegraphics[width=4cm]{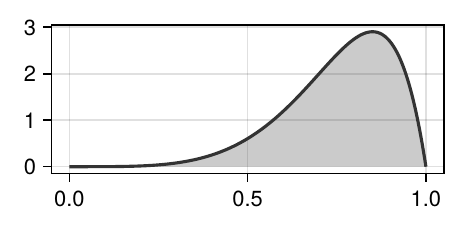}};
        \node[rectangle, below=0cm of cont_dist] (finite) {\includegraphics[width=4cm]{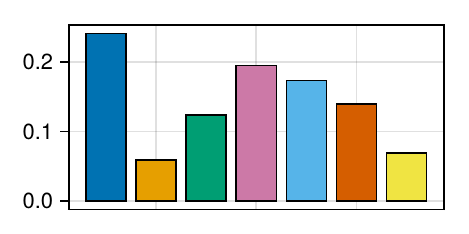}};
        \node[rectangle, right=1.2cm of cont_dist] (hist) {\includegraphics[width=4cm]{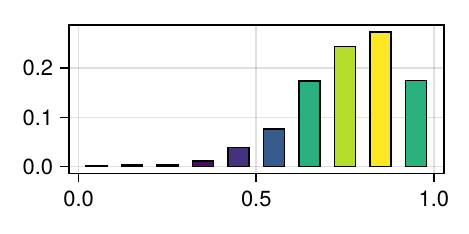}};
        \draw[-{Stealth}] (cont_dist.east) -- (hist) node[midway, above] {$L$ bins};
        \path ($(cont_dist)!0.5!(finite)$) ++(10,0) node[draw] (cat) {$\begin{aligned}
        W(\xi_i,\xi_j) &\mapsto \theta_{ij} \in [0,1]^L \\ X_{ij} &\mapsto \delta_{ij} \in \{0,1\}^L
        \end{aligned}$};
        \draw[-{Stealth}] (hist.east) -| (cat.north);
        \draw[-{Stealth}] (finite.east) -| (cat.south);
    \end{tikzpicture}
    \caption{Schematic representation of the link between a compactly decorated graphon and its finitely decorated graphon approximation. The continuous distribution function (top left) is approximated by a histogram with $L$ bins (top right). It can then be treated similarly to a discrete distribution (bottom left).}
    \label{fig:equivalence_compact_finite}
\end{figure}

Unlike standard density estimation, we must account for the sparsity of real-world networks. In weighted graphs, ``zero'' often signifies the absence of an edge, not just a weight of value zero. We formalize this via the following assumption, which is standard in weighted graph analysis to account for the potential higher prevalence of zero weights \citep{xu_optimal_2020}.

\begin{assumption}
    \label{assumption:DistributionFunction}
    The function $W:[0,1]^2 \rightarrow \measures{[0,1]}$ is such that for all pairs $(x,y) \in [0,1]^2$, the distribution $W(x,y)$ has an atom at $0$. That is, its cumulative distribution function (CDF) can be written as:
    \begin{equation}
        w_{x,y}(t) = (1-p_{x,y})\indicator{t \geq 0} + p_{x,y}\,F_{x,y}(t), 
        \quad \text{for } p_{x,y} \in [0,1],
    \end{equation}
    where $F_{x,y}$ is a distribution function supported on $[0,1]$ with  $F_{x,y}(0)=0$. We drop the dependence on $(x,y)$ when clear from context, writing $w$, $p$, and $F$ instead.
\end{assumption}

\subsection{Constructing the Estimator}

We approximate the CDF $w$ using a piecewise-linear function $\tilde{w}$ defined on $\numK$ equal-sized bins $(x_\indexK,x_{\indexK+1}]$. Specifically, if we let $X \sim w$, we define the parameter vector $\theta \in [0,1]^{\numK+1}$ as:
\begin{equation}
	\theta\indexk{\indexK} =  \begin{cases}
		w(0) & \text{ if } \indexK = 0    \\
		\pr{X \in (x_{l-1},x_{l}]} & \text{ for } l=1,\ldots,L.
    \end{cases}
	\label{eq:theta_hist_from_cont}
\end{equation}
This choice of basis functions allows for both theoretical analysis and computational efficiency. %

The estimation procedure relies on temporarily mapping the continuous problem into the finite framework established in Section~\ref{section:inference_finitely}. We first discretize the observed adjacency matrix $A$ by replacing every entry $A_{ij}$ with an integer label $l \in \{0, \dots, L\}$, corresponding to the bin $(x_{l-1}, x_l]$ in which the weight falls (or 0 for empty edges). This yields an auxiliary finitely decorated graph $\tilde{A}$.

We then compute the estimator $\probshat$ for this auxiliary graph using the minimization defined in eq.~\eqref{eq:least-squares-formulation}. This provides us with an estimated probability matrix $\probshat \in [0,1]^{n \times n \times (L+1)}$, where $\probshat_{ij}$ serves as the weights for our histogram approximation.

Finally, we reconstruct the continuous estimator by using these weights to build the cumulative distribution function. The estimated distribution $\hat{w}_{ij}(x)$ is defined as:
\begin{equation}
    \label{eq:estimator_w_reconstruction}
    \hat{w}_{ij}(x) =  \indicator{x \geq 0} \hat{\theta}_{ij}\indexk{0} + \sum_{\indexK=0}^{\numK-1}\indicator{x \in (x_\indexK,x_{\indexK+1}]}\left(\sum_{k=1}^{\indexK}\hat{\theta}_{ij}\indexk{k}+ \numK_n\left(x-x_\indexK\right)\hat{\theta}_{ij}\indexk{\indexK+1}\right).
\end{equation}
The final graphon estimator $\hat{W}(u,v)$ is the function returning $\hat{w}_{\lceil nu \rceil, \lceil nv \rceil}$.

\subsection{Properties of the Estimator}

To analyze the convergence in the compact setting, we must adapt our smoothness assumptions. Since we are estimating distribution functions, we extend the Hölder condition using the Wasserstein distance $\mathcal{W}_1$ \citep{panaretos_statistical_2019}. The Wasserstein-1 distance is particularly suitable for our framework because here it is equivalent to the $L_1$ norm between cumulative distribution functions \citep{doi:10.1137/1118101} and does not require the existence of probability densities, allowing us to compare distributions with atomic components.

\begin{assumption}
    \label{assumption:Holder_wasserstein}
   $W:[0,1]^2 \rightarrow \measures{[0,1]}$ is Hölder continuous with exponent $\alpha \in (0,1]$. Specifically, $W \in \mathcal{H}(\alpha, M)$ where:
	$$ \sup _{(x, y) \neq\left(x^{\prime}, y^{\prime}\right) \in(0,1)^2} \frac{\mathcal{W}_1\left(W(x, y),W\left(x^{\prime}, y^{\prime}\right)\right)}{\|(x, y)-\left(x^{\prime}, y^{\prime}\right)\|_1^\alpha} \leq M < \infty.$$
\end{assumption}
This assumption is advantageous because it implies the smoothness of the resulting binned graphon approximations and, crucially, does not require the existence of a density function for $W(x,y)$.

\begin{theorem}
	\label{theorem:rate-compact}
	Let $W:[0,1]^2 \rightarrow \measures{[0,1]}$ satisfying Assumption~\ref{assumption:DistributionFunction} and Assumption~\ref{assumption:Holder_wasserstein}. If we select model parameters $k = \left\lceil n^{3/(4\alpha+3)} \right\rceil$, $s=\binom{k+1}{2}$, and number of bins $L = \left\lceil n^{2\alpha/(4\alpha+3)} \right\rceil$, then for all $C'>0$, there exists $C>0$ such that:
	$$
    \frac{1}{n^2}\sum_{i,j}\left(\int_0^1|w_{ij}(x)-\hat{w}_{ij}(x)|\mathrm{d}x\right) \leq Cn^{-2\alpha/(4\alpha+3)}
	$$
    holds with probability at least $1-\exp\left(-C'n\right)$.

    Additionally, if $W(\xi_i,\xi_j)$ admits a density $\mu_{ij}$ that is Lipschitz continuous (bounded by $B$), we obtain the following rate for density estimation:
    $$
        \frac{1}{n^2} \sum_{i,j}\left[\int_0^1\left(\mu_{ij}(x)-\hat{\mu}_{i j}(x)\right)^2 \mathrm{d}x\right] = O_p\left(n^{-4\alpha/(4\alpha+3)}\right).
    $$
\end{theorem}

The first bound is the primary result: it holds for all distributions satisfying Assumptions~\ref{assumption:DistributionFunction}--\ref{assumption:Holder_wasserstein}, including those with atoms and without densities. The Wasserstein-$1$ metric is the natural choice in this generality, since it metrizes weak convergence on $[0,1]$. The second bound requires the stronger assumption that densities exist and are Lipschitz; under this condition, we recover the more familiar $L_2$ rate, which is the square of the Wasserstein-$1$ rate.

The proof of Theorem~\ref{theorem:rate-compact} decomposes the error into four sources: the graphon approximation error (scaling as $k^{-2\alpha}$), the variance of the block-level estimator ($Lk^2/n^2$), the clustering error ($\log(k)/n$), and the discretization error from approximating continuous distributions with $L$ bins ($L^{-1}$). The optimal $(k,L)$ balance the approximation, variance, and discretization terms; the clustering term is automatically dominated under this scaling. The resulting choice of $L$ is in line with the classical Freedman--Diaconis prescription \citep{freedman_histogram_1981} that optimal bin counts grow polynomially in the number of observations per shape.

Compared to the finite case (Theorem~\ref{theorem:holder-rate}), the natural exponent in the rate moves from $2\alpha/(\alpha+1)$ to $4\alpha/(4\alpha+3)$. The two rates are not directly comparable—the finite case bounds a squared $\ell_2$ distance between probability vectors, the continuous case a squared $L_2$ distance between densities—but the continuous setting unavoidably pays for estimating an entire density rather than a finite-dimensional vector. Whether the resulting gap is fundamental or an artifact of our discretization remains an open question.

\section{Numerical experiments}
\label{section:experiments}

\subsection{Empirical Rate of Convergence}

In this section, we showcase the performance of our method by evaluating the Mean Squared Error (MSE) defined in eq.~\eqref{eq:least-squares-formulation} between true and estimated decorated graphons. We consider the three cases described in Table~\ref{tab:sim}: $W_1$ and $W_3$ are Hölder smooth with $\alpha=1$, while $W_2$ is  Hölder continuous with $\alpha = 0.5$. To solve the combinatorial optimization problem necessary to compute our estimators, we use a local greedy search algorithm initialized via spectral clustering, similarly to the approach of Olhede and Wolfe \citep{olhede_network_2014}.

\begin{table}[h]
    \resizebox{\textwidth}{!}{%
	{\begin{tabular}{l|llllll}
			                   &  & \multicolumn{1}{c}{$W_1$} &  & \multicolumn{1}{c}{$W_2$}                &  & \multicolumn{1}{c}{$W_3$}   \\
             \hline 
			$w\indexk{1}(x,y)$ &  & $(1-\min(x,y))(1-|x-y|)$  &  & $1-\sum_{k=2}^{4}w^{(k)}(x,y)$           &  & $\propto 3xy$                                      \\
			$w\indexk{2}(x,y)$ &  & $|x-y|(1-\min(x,y))$      &  & $\sqrt{|x-y|}/2$        &  & $\propto 3\sin(2\pi x)\sin(2\pi y)$                \\
			$w\indexk{3}(x,y)$ &  & $\min(x,y)(1-|x-y|)$      &  & $|\sin(2\pi x)\sin(2\pi y)|/2$              &  & $\propto \exp\left(-3((x-0.5)^2 + (y-0.5)^2)\right)$ \\
			$w\indexk{4}(x,y)$ &  & $\min(x,y)|x-y|$          &  & $\min(x,y)/2$ &  & $\propto 2-3(x+y)$
		\end{tabular}}}
	\caption{Decorated graphon parameters used in the simulations. For $W_3$, the entries in the table are unnormalized: at each $(x,y)$, the four values are passed through a softmax transform $w^{(l)}(x,y) = \exp(g^{(l)}(x,y))/\sum_{l'}\exp(g^{(l')}(x,y))$ to obtain valid probabilities.}
	\label{tab:sim}
\end{table}

We consider $n=\{500,1000,\ldots,5000\}$ and run $10$ simulation trials. More specifically, for a fixed $n$ and $W\in\{W_1,W_2,W_3\}$, for each repetition we draw $\xi_i \sim U[0,1]$ and $A_{ij} \sim W(\xi_i,\xi_j)$. Figure~\ref{fig:sim_rate_cat_both} shows that the errors behave as predicted by our theory (Theorem~\ref{theorem:holder-rate}), and confirms the quality of our estimator both when we a-priori know the block sizes and when we estimate them.

\begin{figure}[h!]
    \centering
    \includegraphics[width=14cm]{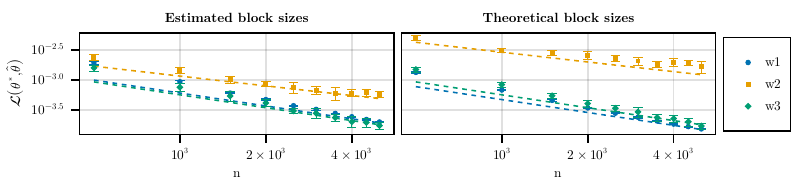}
    \caption{The MSE of our estimator on finitely decorated graphs. Each point represents the average MSE over $10$ independent repetitions, and the error bars represent the standard error. The dashed lines represent the theoretical rates of convergence. Block sizes were estimated  on the left and set to the optimal $\left\lceil n^{\alpha/(\alpha+1)}\right\rceil$ on the right (see Theorem~\ref{theorem:holder-rate}).}
	\label{fig:sim_rate_cat_both}
\end{figure}

Any multiplex network with a finite number of layers can be seen as a finitely decorated graph \citep{kivela_multilayer_2014}. Consider a multiplex network comprising $n$ common nodes across $T$ layers. 
A multiplex network is strictly equivalent to a  $\{0,1\}^T$-decorated graph where $A_{ij}\indexk{t}= A_{ji}\indexk{t} = 1$ if there exists an edge between nodes $i$ and $j$ in layer $t$. As an illustration, we fix $T=2$ and enumerate the different decorations in lexicographic order $\K =  \{[0,0],[1,0],[0,1],[1,1]\}$. 
This is an extension of the multiplex stochastic block model of Barbillon et al. \citep{barbillon_stochastic_2017} in the same way that a graphon extends an SBM\@. With this interpretation, we can see that $W_1$ is the decorated graphon for a 2-layer multiplex network with independent layers and marginal probabilities $\min(x,y)$ and $|x-y|$. 
On the other hand, $W_2$ and $W_3$ encode dependencies across the layers as the probabilities of the bivariate Bernoulli decorations cannot be written as a product of the marginal probabilities. 

\begin{figure}[h!]
	\centering
	\includegraphics[width=14cm]{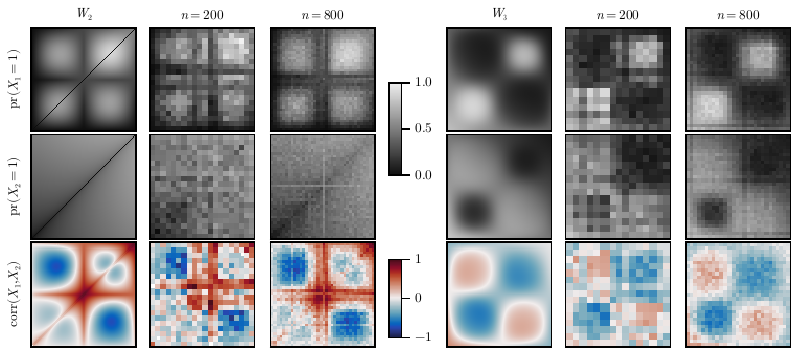}
	\caption{Ground truth $W_2$ and $W_3$ (see \cref{tab:sim}) and estimated graphon with an increasing number of nodes observed. The blocks of the estimator were relabeled to align visually with the ground truth. The difference in smoothness between the two graphons is particularly visible in the correlation between the two layers.}
	\label{fig:sim_picture}
\end{figure}

Figure~\ref{fig:sim_rate_cat_both} shows that our method is equally effective at handling dependent and independent layers. 
For bivariate Bernoulli random variables, we can alternatively parametrize its distribution using the two marginals and the correlation \citep{teugels_representations_1990}.  We use this parametrization to visualize the estimated decorated graphon for different numbers of nodes in Figure~\ref{fig:sim_picture}.

For the compactly decorated case, we look at $\operatorname{Kumaraswamy}(a,b)$ distributions for their wide variety of shapes on the unit interval (similar to Beta distributions, but with faster computational evaluation of the likelihood) \citep{kumaraswamy_generalized_1980}. When testing the second part of Theorem~\ref{theorem:rate-compact}, we restrict $a,b > 1$, ensuring the Lipschitz assumption of the densities. We reuse the same setup as before. Figure~\ref{fig:sim_rate_cont} shows the $L_2$ norm between densities, and the Wasserstein $1$ distance ($\mathcal{W}_1$) between distribution functions. The results are consistent with the theory (see Theorem~\ref{theorem:rate-compact}), which shows that our method can be used to estimate the distribution of the decorations in a compactly decorated graphon.

\begin{figure}[h!]
    \centering
	\includegraphics[width=14cm]{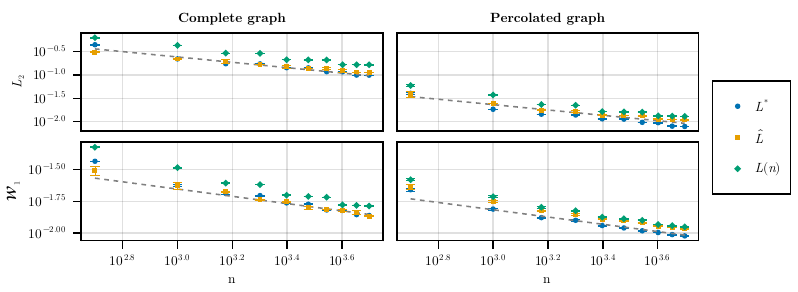}
	\caption{The  error of our estimator for compactly decorated complete graphs (weight of the atom at $0$ is $0$) and percolated (weight of the atom is $1-\log(1-\max(x,y)/2)$). We explore different discretization methods for the interval $[0,1]$: $L^* = n^{2\alpha/(4\alpha+3)}$, $\hat{L}$ is derived from the estimated block sizes and $L(n)=0.4\log(\log(n))^4$ from \citep{xu_optimal_2020}. The dashed lines represent the theoretical rates of convergence from Theorem~\ref{theorem:rate-compact}.}
	\label{fig:sim_rate_cont}
\end{figure}

\section{Illustrative Data Analysis}
\label{section:ada}

We apply our method to a comprehensive dataset of 779 diseases, revealing intricate genetic and symptomatic relationships discussed by Halu et al. \citep{halu_multiplex_2019}. The multiplex network comprises a genotype and phenotype layer where nodes represent diseases. Diseases are linked in the phenotype layer if they share a common symptom, and in the genotype layer, if linked to a common gene. 
The original dataset reports the number of connections as weights in each layer: in the genotype (phenotype) layer, only 130 (330) pairs of diseases have a weight greater than 1 in at least one layer. While the compactly decorated framework introduced accommodates such continuous features, nonparametrically estimating the joint distribution of these weights via a two-dimensional histogram requires a substantially larger sample size than this dataset affords. Consequently, because our primary interest lies in capturing the joint connectivity structure and cross-layer dependencies rather than the granular connection strengths, we binarize these weighted adjacency matrices. By keeping only the presence or absence of links between diseases, we effectively cast the system as a multiplex network, allowing for robust estimation of the cross-layer correlation structure.

The range of the estimated correlation (bottom row of Figure~\ref{fig:diseasome}) being mostly positive aligns with Halu et al.'s finding that diseases with common genetic constituents tend to share symptoms. The clustering seems to identify biologically meaningful disease groups, including autoimmune-inflammatory disorders (e.g., rheumatoid arthritis, multiple sclerosis, and ulcerative colitis in group $1$), hereditary cancer syndromes (e.g., Li-Fraumeni syndrome, familial adenomatous polyposis, and BRCA-associated breast cancer in group $5$), and neuromuscular-cardiovascular conditions (e.g., Duchenne muscular dystrophy, hypertrophic cardiomyopathy, and long QT syndrome in group $8$). While clustering is not our primary goal, this highlights the coherence of our method in uncovering shared mechanisms and relationships within complex disease networks.

\begin{figure}[h!] 
	\centering
	\includegraphics[width=14cm]{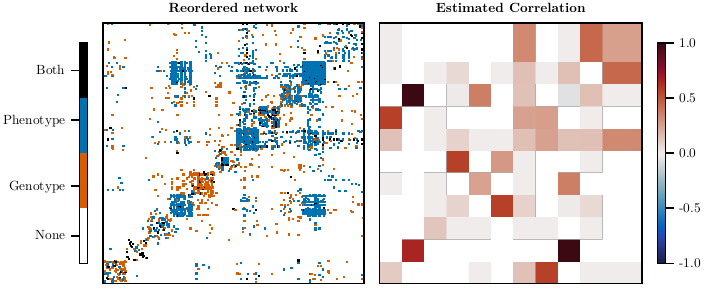}
	\caption{Top: Raw adjacency matrix representing the  Multiplex Network of Human Diseases.  Bottom: An ordered matrix of the network, organized by the refined categorizations derived from our decorated graphon estimation method, illustrating clearer patterns of disease interrelations.}
	\label{fig:diseasome}
\end{figure}

The coherence of these groupings indicates that the estimated block assignments may reflect some latent disease structure. Combined with the positive cross-layer correlation visible in Figure~\ref{fig:diseasome}, this suggests that the estimated decorated graphon $\hat{W}$ captures sufficient dependence between the genotype and phenotype layers to transfer information from one to the other.

We exploit this cross-layer signal by exploring whether phenotypic similarity can predict missing genetic links, a task of practical relevance since identifying shared symptoms is straightforward while uncovering shared genetic underpinnings is more complex and costly. To this end, we emulate a scenario with incomplete data following \citep{gao_rateoptimal_2015}. Rather than introducing new diseases, we randomly mask edges (with a probability of $0.5$) for a random small subset $(5\%)$ of nodes as illustrated in Figure~\ref{fig:link_pred}. We then fit our decorated graphon estimator to this partially observed network. To predict a missing edge, we reintroduce the masked phenotype information and compute the estimated probability of an edge in the genotype layer conditional on the observed phenotype layer. 
Repeating this experiment $100$ times yielded an average Area Under the Curve (AUC) of $0.746 \pm 0.074$. Performing the same experiment with a standard binary graphon estimator (which only considers the genotype layer) resulted in an average AUC of $0.635\pm 0.079$. These results, shown in Figure~\ref{fig:link_pred}, confirm that our estimator successfully captures the correlation structure between layers, allowing it to reconstruct missing genetic interactions using the signal preserved in the phenotypic layer.

\begin{figure}[h!] 
	\centering
	\includegraphics[width=14cm]{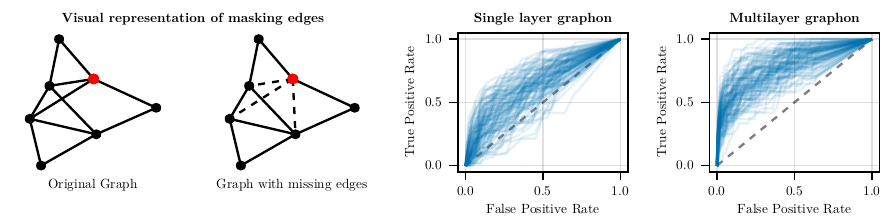}
	\caption{Link prediction experiment on the human disease multiplex network. Left: illustration of the edge masking procedure, where edges incident to a randomly selected node (red) are removed (dashed). Right: ROC curves over $100$ repetitions for predicting the masked genotype edges using a single-layer (genotype only) graphon estimator (AUC $= 0.635 \pm 0.079$) and our multilayer decorated graphon estimator, which conditions on the observed phenotype layer (AUC $= 0.746 \pm 0.074$).}
	\label{fig:link_pred}
\end{figure}

\section{Discussion}
\label{section:discussion}

We have introduced the first nonparametric estimation method for decorated graphons with compactly supported decorations, thus providing a comprehensive approach to infer the underlying generative mechanisms of a broad class of complex networks. Combining a piecewise-constant approximation of the graphon with a one-hot encoding (for finite decorations) or histogram discretization (for continuous decorations), our estimator achieves convergence rates that generalize those of classical binary graphon estimation. The rates decompose transparently into clustering, nonparametric, and (in the continuous case) discretization components, and they apply uniformly across the block (SBM) and shape (SSM) approximations of \citep{gao_rateoptimal_2015,verdeyme_hybrid_2024}.

Several directions remain open. First, while our framework accounts for zero-inflated edge distributions via an atom at zero (Assumption~\ref{assumption:DistributionFunction}), the underlying graphon model~\eqref{eq:aldous-hoover-like} is inherently dense: the expected degree grows linearly in $n$. Extending decorated graphon estimation to the sparse regime, where the expected degree grows sublinearly in $n$, would substantially broaden applicability; recent progress on sparse binary graphons \citep{klopp_oracle_2017,borgs_lp_sparse_I_2019} suggest some possible paths forward. Second, our method currently treats all edges as independent conditional on the latent variables. Incorporating node covariates, as in \citep{chandna_local_2022}, could improve estimation when such information is available. Third, the SSM framework offers flexibility through the shape geometry, but the question of adaptively selecting shapes based on local smoothness of $W$ (rather than using a uniform diameter bound) remains open; see Remark~\ref{remark:appendix_static_shape_number}. Fourth, while our data application demonstrates the benefit of jointly modeling multiple network layers, the sample size required for nonparametric estimation of continuous edge distributions can be demanding. Developing semiparametric alternatives that impose structure on the edge distribution while retaining the nonparametric treatment of the graphon function is a promising direction. Finally, as the size and complexity of network data grow, designing computationally efficient algorithms to scale these procedures remains an interesting and necessary challenge \citep{luo_computational_2024}.

\bibliographystyle{plainnat}
\bibliography{Library.bib}

\clearpage
\appendix

\section{Proofs of the Main Theorems: Finitely Decorated Case}
\label{sec:finitely_decorated}
The proofs extend the arguments of Gao et al.\ \citep{gao_rateoptimal_2015}, Klopp et al.\ \citep{klopp_oracle_2017}, and Verdeyme and Olhede \citep{verdeyme_hybrid_2024} to the decorated setting. The main new difficulty is that, for fixed $i,j$, the coordinates $\obs_{ij}\indexk{1},\ldots,\obs_{ij}\indexk{\numK}$ are dependent even conditionally on the latent variables; this requires replacing Hoeffding's inequality with McDiarmid's in the concentration arguments. The arguments are stated for $(s,k)$-SSM approximations since they cover the SBM as the special case $s=\binom{k+1}{2}$; no step of the proof relies on the additional flexibility of shapes.

The appendix is organized as follows. Appendix~\ref{appendix_a:oracle_ssm} establishes the existence of an oracle estimator within bounded stochastic shape models via a polyomino tiling argument (a polyomino is an edge-connected union of grid squares \citep{golomb_polyominoes_1996}). Appendix~\ref{appendix:transformation} recalls the binary encoding. Appendix~\ref{subsec:proof_shape_rate} proves Theorem~\ref{theorem:shape-rate}, and Appendix~\ref{subsection:holder-rate} proves Theorem~\ref{theorem:holder-rate} (H\"older assumption). Appendix~\ref{sec:oracle} collects the oracle inequalities, Appendix~\ref{sec:compactly_decorated} treats the compactly decorated case, and Appendix~\ref{sec:optimization} discusses practical optimization.

The two-level mapping underlying the $(s,k)$-SSM (and its SBM specialization) is illustrated in Figure~\ref{fig:mapping}: the unit interval is first split into $k$ equal blocks, and the resulting $k\times k$ grid of block pairs is then partitioned into $s\leq\binom{k+1}{2}$ shapes. When $s=\binom{k+1}{2}$, each block pair is its own shape and the construction coincides with a standard $k$-block SBM.

\begin{figure}[h]
    \includegraphics[width=14cm]{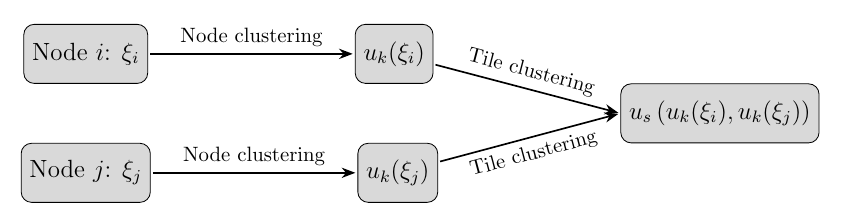}
	\caption{Schematic representation of the mapping from nodes to shape in a $(s,k)$-SSM. Nodes are first aggregated into blocks, which are then further mapped to shapes (image from Verdeyme and Olhede \citep{verdeyme_hybrid_2024}). The standard SBM is recovered when $s=\binom{k+1}{2}$.}
    \label{fig:mapping}
\end{figure}

\subsection{Notation}
Throughout the appendices, $\|\cdot\|$ denotes the Frobenius norm and $\langle\cdot,\cdot\rangle$ the Frobenius inner product unless stated otherwise. Let $\maps_{n,s, k}$ be the set of mappings $\map:[n]^2\to[s]$ of the form $\map(x,y) = u_s(u_k(x),u_k(y))$, where $u_k : [n]\to [k]$ is surjective and $u_s:[k]^2 \to [s]$. Each such $\map$ defines a clustering of node pairs into $s$ groups, each being a union of blocks. We drop the resolution subscript $k$ and write $\maps_{n,s}$ when it is clear from context; one has $|\maps_{n,s}|\leq \max(k,s)^n$ \citep{verdeyme_hybrid_2024}.

\subsection{On the Existence of the Oracle Estimator}
\label{appendix_a:oracle_ssm}

The oracle estimator depends only on a partition of $[0,1]^2$ into disjoint regions, not on the unknown $\probs^*$. For regular blocks, Gao et al.\ \citep{gao_rateoptimal_2015} partition into squares of side $1/k$. For shapes, the geometry of the tessellation is not determined by shape size alone. We therefore restrict to a subclass of stochastic shape models parametrized by resolution $k$ and a diameter bound $d$:

\begin{definition}
	\label{definition:bounded_ssm}
	For $k\in \mathbb{N}$ and $d\in \{1,\ldots,2k\}$, the set $\boundedssm{}$ of \emph{$d$-bounded stochastic shape models with resolution $k$} consists of all $(s,k)$-SSMs (for any $s>0$) in which every pair of blocks assigned to the same shape has Manhattan distance between their centers at most $d/k$.
\end{definition}

\begin{remark}
	This discussion is related to the notion of \emph{elongation} introduced by Verdeyme \citep{verdeyme_nonparametric_2025}, which measures the extent to which shapes are stretched in one direction compared to another.
\end{remark}

Often, we will represent an element of $\boundedssm$ by a map $w:[0,1]\rightarrow \N$, assigning each pair of latents to a shape. Lemma~\ref{lemma:bounded_ssm_error} shows that an oracle estimator can be defined based only on the latent space's geometry.

\begin{lemma}
	\label{lemma:bounded_ssm_error}
	Given $W\in \mathcal{H}(\alpha,M)$ and $n\in \N$, let $\xi \overset{\text{iid}}{\sim} U[0,1]$ and $\theta_{ij}=W(\xi_i,\xi_j)$. For any $k\in [n]$ and $d\in \{1,\ldots,2k\}$, any $w \in \boundedssm$ is such that defining $\map:[n]^2\rightarrow [s]$ as $\map(i,j)=w(\xi_i,\xi_j)$ leads to
	\begin{equation*}
		\frac{1}{n^2}\|\probs - \bar{\probs}(\map^*)\|_F^2 \leq 9M^2\left(\frac{k}{d}\right)^{-2(\alpha\wedge 1)}.
	\end{equation*}
\end{lemma}

\begin{proof}[Proof of\ \ Lemma~\ref{lemma:bounded_ssm_error}]
	Consider a shape containing $(\xi_i,\xi_j)$ and $(\xi_u,\xi_v)$ (i.e., $w(\xi_i,\xi_j)=w(\xi_u,\xi_v)$) and let $a,b$ be their pixels (or block of length $1/k$) with centers $(x_a,y_a)$ and $(x_b,y_b)$ respectively. Then we know that
	\begin{align*}
		|\xi_i-\xi_u| + |\xi_j-\xi_v| &\leq \underbrace{|\xi_i-x_a|}_{\leq 1/2k} + \underbrace{|\xi_j-y_a|}_{\leq 1/2k} + \underbrace{|x_a-x_b| + |y_a-y_b|}_{\leq d/k} +  \underbrace{|\xi_u-x_b|}_{\leq 1/2k} + \underbrace{|\xi_v-y_b|}_{\leq 1/2k} \\ &\leq (d+2)/k.
	\end{align*}
The remainder follows from the H\"older condition. For $(i,j)$ with $\map(i,j)=c$,
	\begin{align*}
		\|\theta_{ij} -\bar{\theta}_{a}(\map) \| & =  \|f(\xi_i,\xi_j) - \bar{\theta}_{c}(\map)\|                                                     \\
		                                         & = \|f(\xi_i,\xi_j) - \frac{1}{n_c}\sum_{\map(u,v)=c}\theta_{uv}\|                                  \\
		                                         & \leq \frac{1}{n_c^*} \sum_{\map(u,v)=c}\|f(\xi_i,\xi_j) -f(\xi_u,\xi_v)\|                          \\
		                                         & \leq \frac{1}{n_c^*} \sum_{\map(u,v)=c}M\left(|\xi_i-\xi_u|+|\xi_j-\xi_v|\right)^{\alpha \wedge 1} \\
		                                         & \leq  M\left(\frac{d+2}{k}\right)^{\alpha\wedge 1}                                                 \\
		                                         & \leq 3 M\left(\frac{d}{k}\right)^{\alpha\wedge 1}.
	\end{align*}
	Squaring and summing over all $(i,j)$ finishes the proof.
\end{proof}

Definition~\ref{definition:bounded_ssm} does not fix the number of shapes. For given $(k,d)$, computing the minimum $s^*$ such that $\boundedssm{}$ contains an element with $s^*$ shapes is NP-hard, since it reduces to tiling the $k\times k$ grid with polyominoes of diameter at most $d$ \citep{demaine_jigsaw_2007}. We therefore derive upper and lower bounds. When $d=0$, each block is its own shape, giving $k(k+1)/2$ shapes. For $d>0$, a shape of diameter $d$ in the Manhattan metric covers at most $(d^2+2d)/2$ blocks when $d$ is even and $(d^2+1)/2$ blocks when $d$ is odd (see Figure~\ref{fig:covering}). Since every element of $\boundedssm$ must tile $[0,1]^2$, the number of shapes $s$ satisfies
$$s\geq \lbshapes{d}{k} = \begin{cases}
		(k^2+k)/2        & \text{ if } d =0               \\
		(k^2+k)/(d^2+1)  & \text{ if } d \text{ is odd}   \\
		(k^2+k)/(d^2+2d) & \text{ if } d \text{ is even}.
	\end{cases}$$

\begin{figure}[h!]
	\centering
	\includegraphics[width=14cm]{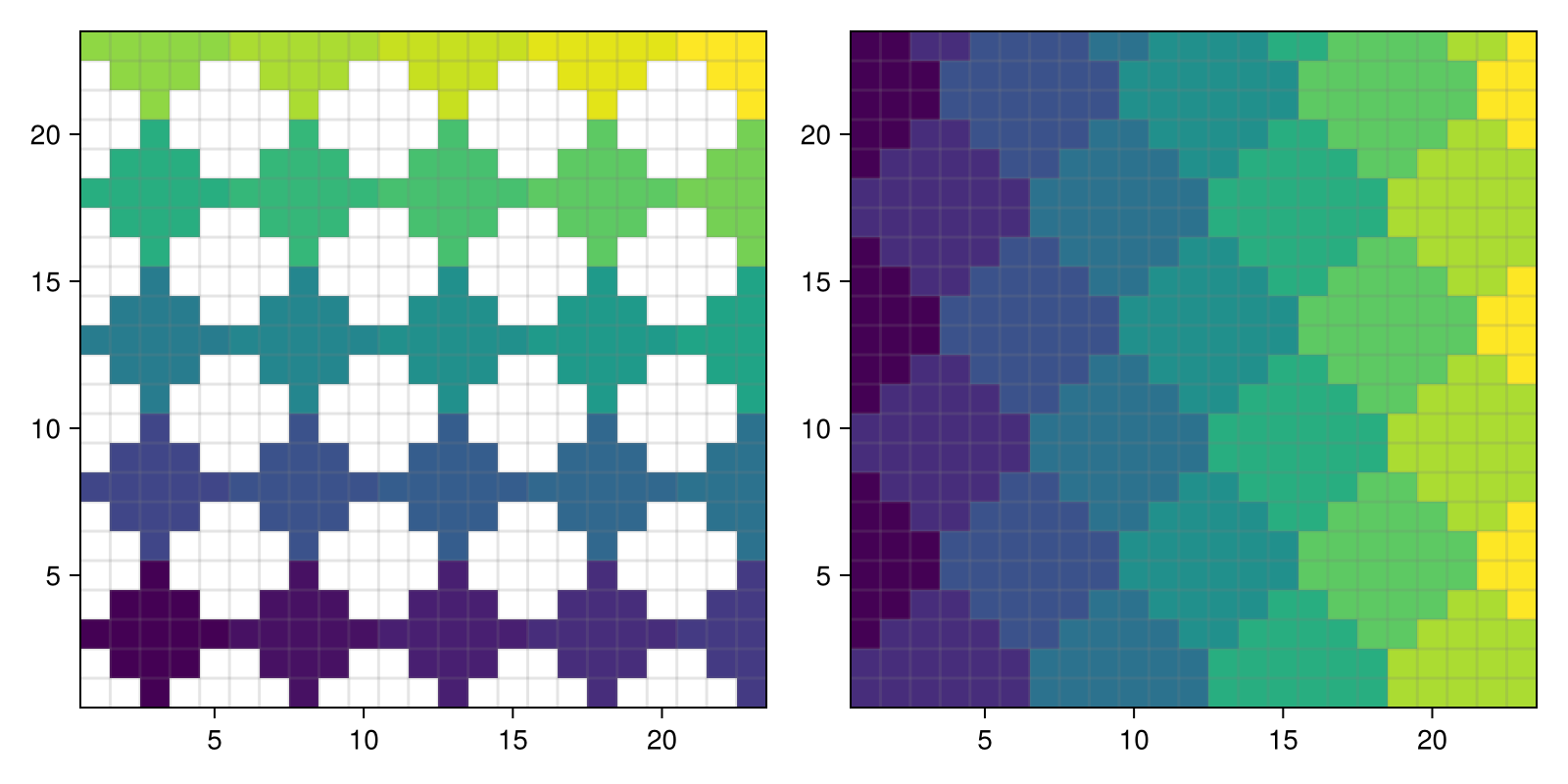}
	\caption{Examples of ensemble of shapes in $\boundedssm[5][23]$ (left) and $\boundedssm[6][23]$ (right). Colored shapes maximize the area they cover, given the diameter constraint.}
	\label{fig:covering}
\end{figure}

To obtain an upper bound on the minimum number of shapes $s^*$, we explicitly construct a covering of $[0,1]^2$ using regular tilings (see Figure~\ref{fig:covering}). When $d\geq k$, five shapes suffice; when $d<k$, the construction yields
\begin{equation*}
	\operatorname{UB}(d,k) \leq \begin{cases}
		(k^2+k)/2                       & \text{ if } d =0               \\
		2k^2/d^2 + 6k/d + 5             & \text{ if } d \text{ is odd}   \\
		2k(k+2)/(d(d+2)) + 5k/(d+2) + 3 & \text{ if } d \text{ is even}.
	\end{cases}
\end{equation*}

For fixed $k\geq 1$ and $1\leq d< k$, there exists an element of $\boundedssm$ with $s^*$ shapes, with  $$9\left(\frac{k^2}{d^2} + \frac{k}{d}+1\right) \geq \ubshapes{d}{k} \geq s^* \geq  \lbshapes{d}{k} \geq \frac{k(k+1)}{d^2+1}.$$
This element is used for building the oracle estimator, and Lemma~\ref{lemma:best_approx_lemma} shows the error this oracle makes.

\begin{lemma}
	\label{lemma:best_approx_lemma}
	Under Assumption~\ref{assumption:holder}, for any $k>1$ and $d\in \{1,\ldots,2k\}$ there exists $\map^*$ representing an element of $\boundedssm{}$ with less than $\ubshapes{d}{k}$ shapes such that
	\begin{equation*}
		\frac{1}{n^2}\|\probs - \bar{\probs}(\map^*)\|_F^2 \leq 9M^2\left(\frac{k}{d}\right)^{-2(\alpha\wedge 1)}.
	\end{equation*}
\end{lemma}

\begin{proof}[Proof of\ \ Lemma~\ref{lemma:best_approx_lemma}]
	Based on Lemma~\ref{lemma:bounded_ssm_error}, we know that any element of $\boundedssm$ satisfies the inequality. We only need to find one with a number of shapes $s^*$ smaller than $\ubshapes{d}{k}$.
\end{proof}

The mapping $\map$ provides an equivalent parametrization of the $(s,k)$-SSM parameter space:
$$\Theta_{s,k}=\left\{\left\{\probs_{i j}\right\} : \probs_{ij}=Q_{\map(i,j)},\; \map \in \maps_{n,s,k} \right\}.$$
For $\eta \in \mathbb{R}^{n\times n \times \numK}$, the shape average is $\bar{\eta}_c(\map) = |\map^{-1}(c)|^{-1}\sum_{(i,j):\map(i,j)=c}\eta_{ij} \in \mathbb{R}^{\numK}$,
where $\map^{-1}(c) = \{(i,j) \in [n]^2 : \map(i,j)=c\}$. The least-squares estimator~\eqref{eq:least-squares-formulation} is then equivalent to finding $\hat{\map}$ and setting $\probshat = \bar{\obs}(\hat{\map})$ \citep{wolfe_nonparametric_2013, gao_rateoptimal_2015, verdeyme_hybrid_2024}.

\begin{remark}
	The $\left(\binom{k}{2},k\right)$-SSM coincides with the stochastic block model with $k$ blocks.
\end{remark}

\subsection{Binary encoding}
\label{appendix:transformation}
We recall the one-hot encoding introduced in Section~\ref{section:inference_finitely} and fix notation for the proofs that follow.
Let $\obs_{ij} \in \{0,1\}^\numK$ with $\obs_{ij}\indexk{l} = 1$ if $\adj_{ij} = x_l \in \K$ and $0$ otherwise. For $p \in [0,1]^\numK$ with $\|p\|_1=1$, we write $X \sim \gbern{p}$ if $\pr{x} =\prod_{\indexK} p_\indexK^{x_\indexK}$ for $x\in \setobs$ and $0$ otherwise. Under eq.~\eqref{eq:aldous-hoover-like}, we have $\obs_{ij} \mid \boldsymbol{\xi} \overset{\text{iid}}{\sim} \gbern{\probs_{ij}}$,
where  $\probs_{ij}\in [0,1]^\numK$  represents $W(\xi_i,\xi_j)$. The key property is that $\obs_{ij}$ is fully determined by its expectation $\E{\obs_{ij}\mid \boldsymbol{\xi}} = \probs_{ij}$.

\subsection{Proof of Theorem~\ref{theorem:shape-rate}, Stochastic Shape Model Assumption}
\label{subsec:proof_shape_rate}

In this case, we denote the true value of each shape by ${Q^*_c} \in [0,1]^{s\times \numK}$ and the oracle assignment by $\map^*\in \maps_{n,s}$ such that $\probs_{ij} = Q^*_{\map(i,j)}$ for any $i\neq j$. For these estimated $\hat{\map}$, we define ${\tilde{Q}_c} \in [0,1]^{s\times \numK}$ by $\tilde{Q}_c = \bar{\probs}_{c}(\hat{\map})$, and $\probstilde_{ij} = \tilde{Q}_{\hat{\map}(i,j)}$ for any $i\neq j$. The diagonal elements of $\probstilde$ are set to the probability distribution on $\K$ that puts all its mass at the zero element.

\begin{proof}[Proof of\ \ Theorem~\ref{theorem:shape-rate}]
	Since $\probshat$ minimizes the least-squares objective, $\|\probshat-\obs\|_F^2 \leq \|\probs - \obs\|_F^2$, which gives
	\begin{equation*}
		\|\probshat-\probs\|_F^2 \leq 2 \langle\probshat-\probs, \obs-\probs\rangle,
	\end{equation*}
	which can be further bounded by
	\begin{equation}
		\label{eq:bounds_ssm}
		\|\probshat-\probs\|^2 \leq 2\|\probshat - \probstilde\|\left|\left\langle\frac{\probshat-\probstilde}{\|\probshat-\probstilde\|}, \obs-\probs\right\rangle\right| + 2\left(\|\probshat - \probstilde\| + \|\probshat-\probs\|\right)\left|\left\langle\frac{\probstilde-\probs}{\|\probstilde-\probs\|}, \obs-\probs\right\rangle\right|.
	\end{equation}
	Lemmas~\ref{lemma:inner_product},\ref{lemma:inner_product_true}, and \ref{lemma:norm} show that for any $C'>0$, there exists a $C>0$ such that all of the following three terms are bounded by $C\sqrt{\numK s + n\log(\max(k,s))}$ with probability at least $1-3\exp\left(-C'n\log(\max(k,s))\right)$:
	\begin{equation*}
		\underbrace{\|\probshat - \probstilde\|\vphantom{\left|\left\langle\frac{\probshat-\probstilde}{\|\probshat-\probstilde\|}\right\rangle\right|} }_{\text{Lemma~\ref{lemma:norm}}}, \quad \underbrace{\left|\left\langle\frac{\probshat-\probstilde}{\|\probshat-\probstilde\|}, \obs-\probs\right\rangle\right|}_{\text{Lemma~\ref{lemma:inner_product}}}, \quad \underbrace{\left|\left\langle\frac{\probstilde-\probs}{\|\probstilde-\probs\|}, \obs-\probs\right\rangle\right|}_{\text{Lemma~\ref{lemma:inner_product_true}}}.
	\end{equation*}
	Combining this bound with eq.~\eqref{eq:bounds_ssm}, we obtain
	\begin{equation*}
		\|\probshat-\probs\|^2 \leq 2C \|\probshat-\probs\|\sqrt{\numK s + n\log(\max(k,s))} + 4C^2\left(\numK s + n\log(\max(k,s))\right).
	\end{equation*}
	Solving for $\|\probshat-\probs\|$ and setting $C_1=(1+\sqrt{5})C$, we get
	\begin{equation*}
		\|\probshat-\probs\| \leq C_1 \sqrt{\numK s + n\log(\max(k,s))},
	\end{equation*}
	with probability at least $1-3\exp\left(-C'n\log(\max(k,s))\right)$. To get the bound in expectation, let $\varepsilon^2 = C_1\left(\numK s/n^2 + \log(\max(k,s))/n\right)$, we have
	\begin{align*}
		\E{n^{-2}\|\probshat-\probs\|^2} & = \E{n^{-2}\|\probshat-\probs\|^2 1_{\{n^{-2}\|\probshat-\probs\|^2 \leq \varepsilon^2\}} } + \E{n^{-2}\|\probshat-\probs\|^2 1_{\{n^{-2}\|\probshat-\probs\|^2 > \varepsilon^2\}}} \\
		                                 & \leq \varepsilon^2 + \mathrm{pr}\left(n^{-2}\|\probshat-\probs\|^2 > \varepsilon^2\right)                                                                                           \\
		                                 & \leq \varepsilon^2 + 3\exp\left(-C'n\log(\max(k,s))\right)                                                                                                                          \\
		                                 & \leq C_1\left(\frac{\numK s}{n^2} + \frac{\log(\max(k,s))}{n}\right) + 3\exp\left(-C'n\log(\max(k,s))\right).
	\end{align*}
	Since $\varepsilon^2$ is the dominating term, this concludes the proof.
\end{proof}

\subsection{Proof of Theorem~\ref{theorem:holder-rate}: Hölder Assumption}
\label{subsection:holder-rate}
\label{appendix:technical}

We remind the reader of Theorem~\ref{theorem:holder-rate}, which extends the results of Verdeyme and Olhede \citep{verdeyme_hybrid_2024} and Gao et al. \citep{gao_rateoptimal_2015} to the finitely decorated graphon case. For ease of reading, we split the statement of the theorem in two parts: Theorem~\ref{theorem:holder-rate-apppendix} deals with the value estimation, while Proposition~\ref{proposition:function_estimation} with the function estimation. The proof of Theorem~\ref{theorem:holder-rate} is then a direct consequence of these two results.

\begin{theorem}
	\label{theorem:holder-rate-apppendix}
	For $W \in \mathcal{H}(\alpha, M)$, $n> \numK$, for any $C^{\prime}>0$, there exists a constant $C>0$ only depending on $C^{\prime}, M, \alpha$, such that the following holds. For any $\Delta\in[0,1]$, set $k=\left\lceil n^{(1+\Delta\alpha)/(\alpha \wedge 1 + 1)} \right\rceil$. Then there exists $s = O\!\left( n^{2/(\alpha \wedge 1+1)}\right)$ such that all shapes have diameter\footnote{For a $(s,k)$-SSM, the diameter of a shape is the maximum $\ell_1$-distance between any two points in that shape region of $[0,1]^2$.} at most $n^{-1/(\alpha \wedge 1 + 1)}$, and
	$$
		\frac{1}{n^2} \sum_{i j}\|\probshat_{i j}-\probs_{i j}\|_2^2 \leq C\left(\numK n^{-2 \alpha /(\alpha+1)}+\frac{\log (n)}{n} \right),
	$$
	with probability at least $1-\exp \left(-C^{\prime} n\right)$. Furthermore,
	$$
		\sup _{W \in \mathcal{H}(\alpha, M)} \E{\frac{1}{n^2} \sum_{i, j \in}\|\probshat_{i j}-\probs_{i j}\|_2^2} \leq C_1\left(\numK n^{-2 \alpha /(\alpha+1)}+\frac{\log (n)}{n}\right),
	$$
	for some other constant $C_1>0$ only depending on $M$. Both the probability and the expectation are jointly over $\left\{\adj_{i j}\right\}$ and $\left\{\xi_i\right\}$.
\end{theorem}

\begin{proof}[Proof of\ \ Theorem~\ref{theorem:holder-rate-apppendix}]
The strategy is to combine the oracle inequality from Theorem~\ref{theorem:shape-rate} with the approximation bound from Lemma~\ref{lemma:best_approx_lemma}, and then optimize the resolution parameters $k$ and $d$.

We set $\map^*$ to the oracle from Lemma~\ref{lemma:best_approx_lemma}, and define $Q^*_c = \bar{\probs}_{c}(\map^*)$, $\probs^*_{ij} = Q^*_{\map(i,j)}$ for $i\neq j$, and $\probs^*_{ii}=\delta_{0_{\K}}$.
	The argument follows Verdeyme and Olhede \citep[Proof of Theorem 3.3]{verdeyme_hybrid_2024}; we detail the final optimization step. Set $k=n^\delta$ and $d=n^\beta$ for $\delta,\beta \in [0,1]$, so that shapes have diameter $d/k$. For all $C'>0$, there exists $C_1>0$ such that
	\begin{equation}
		\label{eq:bound_before_delta}
		\frac{1}{n^2}\|\probshat-\probs\|_F^2 \leq C_1\left(\left(\frac{k}{d}\right)^{-2(\alpha \wedge 1)} + \frac{Ls^*}{n^2} + \frac{\log(\max(k,s^*))}{n}\right),
	\end{equation}
	with probability at least $1-\exp(-C'n)$.
	Based on Lemma~\ref{lemma:best_approx_lemma}, we have that $s^*$ can be bounded, such that for $\delta>\beta$, for $C_2>0$ we obtain
	\begin{align*}
		\frac{1}{n^2}\|\probshat-\probs\|_F^2 & \leq C_2\left(\left(\frac{k}{d}\right)^{-2(\alpha \wedge 1)} + \frac{k^2}{d^2}\frac{L}{n^2} + \frac{\log(\max(k,k^2/d^2))}{n}\right) \\
		                                      & \leq 2C_2\left(n^{-2(\delta-\beta)(\alpha \wedge 1)} + Ln^{-2+2(\delta-\beta)}+ \frac{\log(n)}{n}\right).
	\end{align*}

	For any pair $\delta,\beta$ such that $\delta-\beta=1/(\alpha+1)$, we obtain the rates
	\begin{equation*}
		\begin{cases}
			n^{-2\alpha/(\alpha+1)} & \text{ if } \alpha < 1 \\
			\log(n)/n               & \text{ else}.
		\end{cases}
	\end{equation*}
    The diameter of the shapes is upper bounded by $d/k = n^{\beta-\delta} = n^{-1/(\alpha+1)}$. This concludes the proof.
\end{proof}

\begin{remark}
    \label{remark:appendix_static_shape_number}
	For any optimal $\beta,\delta$ as above, the minimal number of shapes in $\boundedssm$ is of order $n^{2(\delta-\beta)}=n^{2\alpha/(\alpha+1)}$ (see \cref{appendix_a:oracle_ssm}). With a uniform diameter constraint, this yields only a linear reduction in parameters compared to the standard block model. Allowing shapes of different diameters, adapted to the local smoothness of $W$, could improve this. A natural extension would be to let each block $(i,j)$ choose its own shape diameter $d_{ij}$, smaller where $W$ varies rapidly and larger in smoother regions. This is analogous to spatially adaptive histogram estimators \citep{lepski_optimal_1997} and variable-bandwidth kernel methods, where the bandwidth is selected locally. Such an approach would reduce the total number of shapes needed in regions where $W$ is smooth, while preserving resolution where it is not, potentially leading to improved rates that adapt to spatially inhomogeneous regularity. We leave this direction for future work.
\end{remark}

\begin{proposition}
	\label{proposition:function_estimation}
	For $W \in \mathcal{H}(\alpha, M)$, $n>\numK$, for any $C^{\prime}>0$, there exists a constant $C>0$ only depending on $C^{\prime}, M, \alpha$ and $\beta$, where the following holds
	$$
		\operatorname{MISE}\left(\widehat{W}_{\probshat}, W\right) \leq C\left(\numK n^{-2 \alpha /(\alpha+1)}+\frac{\log (n)}{n}  + n^{-\alpha \wedge 1}\right),
	$$
	with probability at least $1-\exp \left(-C^{\prime} n\right)$, with $s$ and $k$ as in Theorem~\ref{theorem:holder-rate}. Furthermore,
	$$
		\sup _{W \in \mathcal{H}(\alpha, M)} \E{\operatorname{MISE}\left(\widehat{W}_{\probshat}, W\right)} \leq C_1\left(\numK n^{-2 \alpha /(\alpha+1)}+\frac{\log (n)}{n} + n ^{-\alpha\wedge 1}\right),
	$$
	for some other constant $C_1>0$ only depending on $M$. Both the probability and the expectation are jointly over $\left\{\adj_{i j}\right\}$ and $\left\{\xi_i\right\}$.
\end{proposition}

\begin{proof}[Proof of\ \ Proposition~\ref{proposition:function_estimation}]
	Klopp et al. \citep{klopp_oracle_2017} showed that the mean integrated error is bounded by

	\begin{equation*}
		\E{\mathrm{MISE}\left(\widehat{W}_{\probshat}, W\right)} \leq 2 \underbrace{\E{\frac{1}{n^2}\left\|\probshat-\probs\right\|_F^{2}}}_{\text{estimation error}} + 2\underbrace{\vphantom{\E{\frac{1}{n^2}\left\|\probshat-\probs\right\|_F^{2}}}\E{\mathrm{MISE}\left(W_{\probs},W\right)}}_{\text{agnostic error}},
	\end{equation*}
	where the agnostic error is the distance between the true graphon and its discretized version sampled at the unknown $\{\xi_i\}$. A direct adaptation of Klopp et al. \citep[Prop. 3.5]{klopp_oracle_2017} shows that the agnostic error  is bounded by $n^{-\alpha\wedge 1}$, which leads to
	\begin{equation}
		\sup _{W \in \mathcal{H}(\alpha, M)} \E{\operatorname{MISE}\left(\widehat{W}_{\probshat}, W\right)} \leq C_2\left( n^{-2 \alpha /(\alpha+1)}+\frac{\log (n)}{n} + n ^{-\alpha\wedge 1}\right),
	\end{equation}
	where we use Theorem~\ref{theorem:holder-rate} to bound the estimation error.
\end{proof}

\section{Oracle Inequalities}
\label{sec:oracle}

This section establishes the three concentration bounds used in the proof of Theorem~\ref{theorem:shape-rate}: Lemmas~\ref{lemma:inner_product} and~\ref{lemma:inner_product_true} control the inner-product terms in the decomposition~\eqref{eq:bounds_ssm}, and Lemma~\ref{lemma:norm} controls the norm term. The main technical tool is Lemma~\ref{lemma:covering_number_proba}, which adapts the chaining argument of Verdeyme and Olhede \citep[Lemma~A.1]{verdeyme_hybrid_2024} to account for the dependence across decoration coordinates via McDiarmid's inequality.

\begin{lemma}
	\label{lemma:inner_product}
	For any constant $C'>0$, there exists a constant $C>0$ only depending on $C'$ such that
	$$\left|\left\langle\frac{\probshat-\probstilde}{\|\probshat-\probstilde\|}, \obs-\probs\right\rangle_F\right| \leq C\sqrt{\numK s + n\log(\max(k,s))},$$

	with probability at least $1-\exp\left(-C'n\log(\max(k,s))\right).$
\end{lemma}

\begin{proof}[Proof of\ \ Lemma~\ref{lemma:inner_product}]
	For each $\map \in \maps_{n, s}$, define the set $\mathcal{B}_{\map}$ by $\mathcal{B}_{\map}=\left\{\left\{a_{i j}\right\}: a_{i j l}=Q_{cl}\right.$ if $(i, j) \in \map^{-1}(c)$ for some $Q_{cl}$, and $\left.\sum_{i j l} a_{i j l}^2 \leq 1\right\}$. In other words, $\mathcal{B}_{\map}$ collects the element of $\mathcal{B}$ as defined in Lemma~\ref{lemma:covering_number_proba} determined by $w$. We then get
	$$
		\left|\left\langle\frac{\probshat-\probstilde}{\|\probshat-\probstilde\|}, \obs-\probs\right\rangle_F\right| \leq \max _{\map \in \maps_{n, s}} \sup _{a \in \mathcal{B}_{\map}} \left|\sum_{\indexK}\sum_{i j} a_{i j \indexK}\left(\obs_{i j}\indexk{\indexK}-\probs_{i j}\indexk{\indexK}\right)\right|.
	$$

	Using a union bound argument and Lemma~\ref{lemma:covering_number_proba}, we get
	\begin{align*}
		\mathrm{pr}\left(\max _{\map \in \maps_{n, s}} \sup _{a \in \mathcal{B}_{\map}} \left|\sum_{\indexK}\sum_{i j} a_{i j \indexK}\left(\obs_{i j}\indexk{\indexK}-\probs_{i j}\indexk{\indexK}\right)\right|>t\right)  \leq \sum_{\map \in \maps_{n,s}} \mathcal{N}\left(1/2,\mathcal{B}_{\map}, \|\cdot\|\right)\exp(-t^2/8).
	\end{align*}

	Since $\mathcal{B}_{\map}$ has $(\numK-1) s$ degree of freedom, a standard bound for covering numbers implies  $\mathcal{N}\left(1/2,\mathcal{B}_{\map}, \|\cdot\|\right) \leq \exp(C_1\numK s)$ \citep[Lemma 4.1]{pollard_empirical_1990}. Using $\left|\maps_{n,s}\right| < \exp(n \log(\max(k,s)))$, we get

	$$\mathrm{pr}\left(\max _{\map \in \maps_{n, s}} \sup _{a \in \mathcal{B}_{\map}} \left|\sum_{\indexK}\sum_{i j} a_{i j \indexK}\left(\obs_{i j}\indexk{\indexK}-\probs_{i j}\indexk{\indexK}\right)\right|>t\right)\leq \exp\left(-t^2/8 + C_1 \numK s + n \log(\max(k,s))\right).$$

	Picking $t^2 \propto \numK s +n \log(\max(k,s))$ finishes the proof.
\end{proof}

\begin{lemma}
	\label{lemma:inner_product_true}
	For any constant $C'>0$, there exists a constant $C>0$ only depending on $C'$, such that
	$$\left|\left\langle\frac{\probstilde-\probs}{\|\probstilde-\probs\|}, \obs-\probs\right\rangle_F\right| \leq C\sqrt{n\log(\max(k,s))},$$
	with probability at least $1-\exp\left(-C'n\log(\max(k,s))\right)$.
\end{lemma}

\begin{proof}[Proof of\ \ Lemma~\ref{lemma:inner_product_true}]
	Follows by using a union bound argument on $\maps_{n,s}$ and Lemma~\ref{lemma:covering_number_proba} (as in the proof of Lemma~\ref{lemma:inner_product}).
\end{proof}

\begin{lemma}
	\label{lemma:norm}
	For any constant $C'>0$, there exists a constant $C>0$ only depending on $C'$, such that
	$$\|\probshat-\probstilde\|_F \leq C\sqrt{\numK s + n \log\left(\max(k,s)\right)}$$

	with probability at least $1-\exp\left(-C'n\log\left(\max(k,s)\right)\right)$.
\end{lemma}

\begin{proof}[Proof of\ \ Lemma~\ref{lemma:norm}]

	\begin{align*}
		\|\probshat-\probstilde\|_F^2
		 & = \sum_{l=1}^\numK\|\probshat\indexk{l}-\probstilde\indexk{l}\|_2^2                                                                                                            \\
		 & = \sum_{l=1}^\numK\sum_{i,j=1}^{n}\left(\probshat\indexk{l}_{ij}-\probstilde\indexk{l}_{ij}\right)^2                                                                           \\
		 & = \sum_{c \in [S]}\sum_{l=1}^\numK |\hat{\map}^{-1}(c)|\left(\bar{\obs}\indexk{l}_{c}(\hat{\map}) - \bar{\probs}\indexk{l}_{c}(\hat{\map})\right)^2        \\
		 & \leq \max_{\map \in \maps_{n,s}}\sum_{c \in [S]}\sum_{l=1}^l \left|\map^{-1}(c)\right|\left(\bar{\obs}\indexk{l}_{c}(\map) - \bar{\probs}\indexk{l}_{c}(\map)\right)^2
	\end{align*}

	For a given $\map \in \maps_{n,s}$, let $n_c = |\map^{-1}(c)|$ and define
	$$ V_{c}(\map)  =  \frac{1}{n_c}\sum_{\indexK}\left(\sum_{(i,j) \in \map^{-1}(c)}\left(\obs_{ij}\indexk{\indexK}-\probs_{ij}\indexk{\indexK}\right)\right)^2.$$

	We then have
	\begin{equation*}
		\|\probshat-\probstilde\|_F^2 \leq \max _{\map \in \maps_{n, s}} \sum_{c \in[s]} \E{V_c(\map)}+\max _{\map \in \maps_{n, s}} \sum_{c \in[s]}\left(V_c(\map)-\E{V_c(\map)}\right).
	\end{equation*}

	We bound the first term:
	\begin{align*}
		\E{V_c(\map)} & = \E{\frac{1}{n_c}\sum_{\indexK}\left(\sum_{(i,j) \in \map^{-1}(c)}\left(\obs_{ij}\indexk{\indexK}-\probs_{ij}\indexk{\indexK}\right)\right)^2} \\
		              & = \frac{1}{n_c}\sum_{\indexK}\E{\left(\sum_{(i,j)\in \map^{-1}(c)}\left(\obs_{ij}\indexk{\indexK}-\probs_{ij}\indexk{\indexK}\right)\right)^2}  \\
		              & = \frac{1}{n_c}\sum_{\indexK}\sum_{(i,j)\in \map^{-1}(c)}\mathrm{var}\left(\obs_{ij}\indexk{\indexK}\right)                                     \\
		              & \leq \numK,
	\end{align*}
	as for a fixed $\indexK$, the variables $\obs_{ij}\indexk{\indexK}$ and $\obs_{qr}\indexk{\indexK}$ are independent for all $(q,r) \neq (i,j)$ and $(q,r)\neq (j,i)$, and $\mathrm{var}\left(\obs_{ij}\indexk{\indexK}\right) \leq 1$.

	We now show that for a fixed $w$, $V_c(\map)$ is a sub-exponential random variable with constant sub-exponential parameter. This will allow us to finish the proof as in Verdeyme and Olhede \citep{verdeyme_hybrid_2024}, replacing the upper bound of $\E{V_c(\map)}$ by the one computed above.
	We first have \begin{equation*}
		V_{c}(\map) \leq \frac{1}{n_c}\sum_{\indexK}\left(\sum_{(i,j) \in \map^{-1}(c)}\left|\obs_{ij}\indexk{\indexK}-\probs_{ij}\indexk{\indexK}\right|\right)^2 \leq \frac{1}{n_c}\left(\sum_{\indexK}\sum_{(i,j) \in \map^{-1}(c)}\left|\obs_{ij}\indexk{\indexK}-\probs_{ij}\indexk{\indexK}\right|\right)^2.
	\end{equation*}
	We then have, for any $t>0$
	\begin{align*}
		\mathrm{pr}\left(V_{c}(\map)>t\right) & \leq \mathrm{pr}\left( \frac{1}{n_c}\left(\sum_{\indexK}\sum_{(i,j) \in \map^{-1}(c)}\left|\obs_{ij}\indexk{\indexK}-\probs_{ij}\indexk{\indexK}\right|\right)^2 > t\right) \\
		                                      & = \mathrm{pr}\left(\sum_{\indexK}\sum_{(i,j) \in \map^{-1}(c)}\left|\obs_{ij}\indexk{\indexK}-\probs_{ij}\indexk{\indexK}\right| > \sqrt{t n_c}\right)                      \\
		                                      & = \mathrm{pr}\left(\sum_{(i,j) \in \map^{-1}(c)}\left\|\obs_{ij}-\probs_{ij}\right\|_1 > \sqrt{t n_c}\right).
	\end{align*}

	Since $\left\|\obs_{ij}-\probs_{ij}\right\|_1 \leq \|\obs_{ij}\|_1 + \|\probs_{ij}\|_1 = 2$, we get that $\left\|\obs_{ij}-\probs_{ij}\right\|_1$ is a sub-Gaussian random variable. Using Hoeffding's inequality for sub-Gaussian variable \citep[Prop. 5.10]{vershynin_introduction_2011}, we have
	\begin{align*}
		\mathrm{pr}\left(V_{c}(\map)>t\right) \leq \exp\left(1- C\frac{t n_c}{4 n_c}\right) \leq \exp\left(1- Ct/4\right)
	\end{align*}
	for some universal constant $C>0$.

	The proof then concludes exactly as in Verdeyme and Olhede \citep[proof of Lemma A.4]{verdeyme_hybrid_2024} replacing $s$ by $s\numK$ in the upper bound of $\sum_c\E{V_c(\map)}$.
\end{proof}

\subsection{Auxiliary Result}

\begin{lemma}
	\label{lemma:covering_number_proba}
	Let $\mathcal{B} \subset\left\{a \in \mathbb{R}^{n \times n \times \numK}: \|a\|_F \leq 1\right\}$.
	Then for any $a \in \mathcal{B}$ we have

	$$\mathrm{pr}\left(\left|\sum_{\indexK}\sum_{i j} a_{i j \indexK}\left(\obs_{i j}\indexk{\indexK}-\probs_{i j}\indexk{\indexK}\right)\right|>t\right) \leq \exp \left(- t^2/8\right).$$

	If we additionally suppose that $\mathcal{B}$ is such that for any $a,b \in \mathcal{B}$
	\begin{equation}
		\label{eq:condition_covering}
		\frac{a-b}{\|a-b\|_F} \in \mathcal{B},
	\end{equation}
	we then have
	$$
		\mathrm{pr}\left(\sup _{a \in \mathcal{B}}\left|\sum_{\indexK}\sum_{i j} a_{i j \indexK}\left(\obs_{i j}\indexk{\indexK}-\probs_{i j}\indexk{\indexK}\right)\right|>t\right) \leq \mathcal{N}(1 / 2, \mathcal{B},\|\cdot\|) \exp \left(-C t^2/32\right).
	$$
\end{lemma}

\begin{proof}[Proof of\ \ Lemma~\ref{lemma:covering_number_proba}]

	Let us first notice that
	$$\left|\sum_{\indexK}\sum_{i j} a_{i j \indexK}\left(\obs_{i j}\indexk{\indexK}-\probs_{i j}\indexk{\indexK}\right)\right| = |\langle a, \obs-\probs\rangle|.$$

	We will show that for each $b \in \mathcal{B}$, $\langle b, \obs-\probs \rangle$ is a function of $\{\obs_{ij}\}_{i > j}$ that satisfies the bounded difference properties.
	\begin{align*}
		\langle b, \obs-\probs \rangle = \sum_{\indexK}\sum_{i j} b_{i j \indexK}\left(\obs_{i j}\indexk{\indexK}-\probs_{i j}\indexk{\indexK}\right) = \sum_{i j} \langle b_{ij}, \obs_{ij} \rangle - 2\sum_{i j} \langle b_{ij}, \probs_{ij} \rangle,
	\end{align*}

	where $b_{ij} = (b_{ij1},\ldots,b_{ij\numK})$. Let us examine what happens when we change the value of this function's $(k,l)$th coordinate. The vector $\obs_{kl}$ is binary with exactly one non-zero component $r \in [\numK]$. Its contribution to the function value is $b_{ijr} + b_{jir}$. Changing the value of $\obs_{kl}$ will then change this contribution to $b_{ijq} + b_{jiq}$ for a $q\in [\numK]$. The absolute value of the difference is then
	$$|b_{ijr} + b_{jir} -b_{ijq} - {b_{jiq}}|\leq 4 \max_{\indexK}|b_{ij\indexK}|.$$

	Notice that $\E{\langle b, \obs-\probs \rangle}=0$. Using the bounded difference property we have just shown, we use McDiarmid's inequality to get
	\begin{equation}
		\label{eq:fixed_b_inner_product_supp}
		\mathrm{pr}\left(\left|\langle b, \obs-\probs \rangle\right| \geq \epsilon \right) \leq \exp\left(-\frac{2\epsilon^2}{16\sum_{ij}\max_{\indexK}(b_{ij\indexK})^2}\right) \leq  \exp\left(-\frac{\epsilon^2}{8}\right),
	\end{equation}

	since $\sum_{ij}\max_{\indexK}(b_{ij\indexK})^2 \leq \sum_{ijl}(b_{ij\indexK})^2 = \|b\|_F^2 \leq 1$. This shows the first part of the proposition.

	Let $\mathcal{B}'$ be a $1/2$-net of $\mathcal{B}$ such that $|\mathcal{B}'| \leq \mathcal{N}(1/2, \mathcal{B},\|\cdot\|)$. For any $a \in \mathcal{B}$ there is a $b \in \mathcal{B}'$ such that $\|a-b\| \leq 1/2$.

	Thus,
	$$
		\begin{aligned}
			|\langle a, \obs-\probs\rangle| & \leq|\langle a-b, \obs-\probs\rangle|+|\langle b, \obs-\probs\rangle|                                             \\
			                                & \leq\|a-b\|\left|\left\langle\frac{a-b}{\|a-b\|}, \obs-\probs\right\rangle\right|+|\langle b, \obs-\probs\rangle| \\
			                                & \leq \frac{1}{2} \sup _{a \in \mathcal{B}}|\langle a, \obs-\probs\rangle|+|\langle b, \obs-\probs\rangle|,
		\end{aligned}
	$$
	where the last inequality follows from the assumption eq.~\eqref{eq:condition_covering}.
	Taking the supremum with respect to $\mathcal{B}$ and maximum with respect to $\mathcal{B}^{\prime}$ on both sides, we have
	$$
		\sup _{a \in \mathcal{B}}\left|\langle a,\obs-\probs\rangle\right| \leq 2 \max _{b \in \mathcal{B}^{\prime}}\left|\langle b, \obs-\probs \rangle\right|.
	$$

	We then get
	\begin{align*}
		\mathrm{pr}\left(\sup _{a \subset \mathcal{B}}\left|\langle a,\obs-\probs\rangle\right|\geq t\right) & \leq \mathrm{pr}\left(2 \max _{b \in \mathcal{B}^{\prime}}\left|\langle b, \obs-\probs \rangle\right| \geq t\right) \\
		                                                                                                     & \leq \sum_{b \in \mathcal{B}'}\mathrm{pr}\left(\left|\langle b, \obs-\probs \rangle\right| \geq \frac{t}{2}\right).
	\end{align*}

	Combining what we have with eq.~\eqref{eq:fixed_b_inner_product_supp}, we get
	\begin{align*}
		\mathrm{pr}\left(\sup _{a \in \mathcal{B}}|\langle a, \obs-\probs\rangle|>t\right) & \leq \sum_{b \in \mathcal{B}'}\exp\left(-t^2/32\right)                \\
		                                                                                   & = |\mathcal{B}'|\exp\left(-t^2/32\right)                              \\
		                                                                                   & \leq \mathcal{N}(1/2, \mathcal{B},\|\cdot\|)\exp\left(-t^2/32\right),
	\end{align*}
	which concludes the proof.
\end{proof}

\section{Proofs of the Main Theorems: Compactly Decorated Case}
\label{sec:compactly_decorated}

The strategy is to discretize the continuous decoration space into $\numK$ bins, reduce to the finitely decorated setting of Appendix~\ref{sec:finitely_decorated}, and then control the discretization error.

We approximate $w$ by $\tilde{w}$, defined on $\numK$ equal-sized bins $[0,1]=\cup[x_\indexK,x_{\indexK+1}]$ with $x_\indexK = \indexK/\numK$ for $\indexK=0,\ldots,\numK-1$. Including the atom at $0$, $\tilde{w}$ is parametrized by $\numK+1$ values: the atom weight and the $\numK$ bin heights. Let $\histmeasures{[0,1]}$ denote the set of all such distributions (see Figure~\ref{fig:discretization}).

\begin{figure}[h!]
	\centering
	\includegraphics[width=14cm]{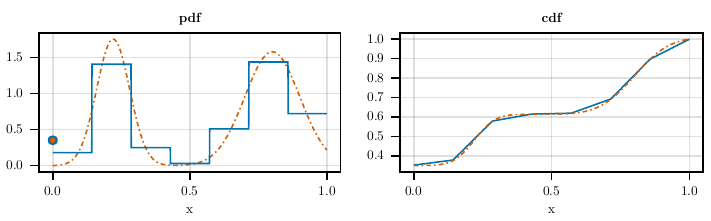}
	\caption{Approximation of a mass/distribution function by a binned distribution with $7$ bins (eq.~\eqref{eq:def_w_tilde}).}
    \label{fig:discretization}
\end{figure}

We define $\tilde{w} \in \histmeasures[L]{[0,1]}$ as the discrete approximation of $w$ by the following distribution function
\begin{equation}
	\label{eq:def_w_tilde}
    \tilde{w}(x) = \indicator{x \geq 0}(1-p)+ p\sum_{\indexK=0}^{\numK-1}\indicator{x \in (x_\indexK,x_{\indexK+1}]}\left[\cdf{x_\indexK} + \numK\left(x-x_\indexK\right)\left( \cdf{x_{\indexK+1}}- \cdf{x_{\indexK}}\right)\right].
\end{equation}

\subsection{Wasserstein $1$ distance}
\label{subsection:proof_thm_rate_compact}
In this section, we will represent a probability measure by its distribution function. To compare $w_1,w_2 \in \measures{[0,1]}$, two probability distributions on $[0,1]$, we will use the $1$-Wasserstein distance, which simplifies to \citep{dallaglio_sugli_1956}

$$\Wd[1]{w_1, w_2}=\int_{[0,1]}\left|w_1(x)-w_2(x)\right| \mathrm{d}x.$$

\begin{lemma} Under Assumption~\ref{assumption:DistributionFunction}, let $\tilde{w}_{ij} \in \histmeasures[L]{[0,1]}$ (and related $\tilde{\theta}_{ij}$) be the discretization described in eq.~\eqref{eq:theta_hist_from_cont}. For all $\hat{w}_{ij} \in \histmeasures[L]{[0,1]}$ (and related $\hat{\theta}_{ij}$),
	\begin{equation*}
		\frac{1}{n^2}\sum_{i,j}\mathcal{W}_1\left(w_{ij}, \hat{w}_{ij}\right)  \leq 2\numK^{-1} +  \frac{2}{n^2}\|\tilde{\theta}-\hat{\theta}\|_1
	\end{equation*}
	\label{lemma:w1_dist}
\end{lemma}

\begin{proof}[Proof of\ \ Lemma~\ref{lemma:w1_dist}]

	We first look at the discretization error, or going from $w$ to $\tilde{w}$:
	\begin{align*}
		\mathcal{W}_1\left(w, \tilde{w}\right) & = \int_{[0,1]}\left| \delta_x([0,\infty))(1-p) + pF(x)-  \delta_x([0,\infty))(1-p) - p\tilde{F}(x)\right| \mathrm{d}x                         \\
		                                       & =\int_{[0,1]}p\left|F(x)-\tilde{F}(x)\right| \mathrm{d}x                                                                                      \\
		                                       & \leq \sum_{\indexK=0}^{\numK-1}\int_{x_{l}}^{x_{l+1}}\abs{\cdf{x_{l}} + \numK(x-x_{l})\left(\cdf{x_{l+1}}-\cdf{x_{l}}\right)-\cdf{x}}dx       \\
		                                       & \leq \sum_{\indexK=1}^{\numK}\int_{x_{l}}^{x_{l+1}}\abs{\numK(x-x_{l})\left(\cdf{x_{l+1}}-\cdf{x_{l}}\right)} + \abs{\cdf{x_{l}}  -\cdf{x}}dx \\
		                                       & \leq \sum_{\indexK=1}^{\numK}\int_{x_{l}}^{x_{l+1}}\abs{\cdf{x_{l}}-\cdf{x_{l+1}}} + \abs{\cdf{x_{l}}  -\cdf{x_{l+1}}}dx                      \\
		                                       & \leq \sum_{\indexK=1}^{\numK} 2(\cdf{x_{l+1}}-\cdf{x_{l}})/\numK                                                                              \\
		                                       & \leq 2/\numK.
	\end{align*}
    Now we are interested in the difference between two histogram distributions on the same bins $w_1,w_2 \in \histmeasures{[0,1]}$. We have, for $x \in [x_\indexK,x_{\indexK+1}]$,
	\begin{align*}
		\abs{w_1(x)-w_2(x)} & = \abs{\theta_1\indexk{0}-\theta_2\indexk{0} + \sum_{k=1}^{\indexK}\left(\theta_1\indexk{k}-\theta_2\indexk{k}\right)+ \numK\left(x-x_\indexK\right)\left(\theta_1\indexk{\indexK+1}-\theta_2\indexk{\indexK+1}\right)},
	\end{align*}
	thus, since $\int_{x_\indexK}^{x_{\indexK+1}}\left(x-x_\indexK\right)\mathrm{d}x=L^{-2}$,
	\begin{align*}
		\mathcal{W}_1\left(w_1, w_2\right) & \leq  \sum_{\indexK=0}^{\numK-1}\int_{x_\indexK}^{x_{\indexK+1}}\left(\sum_{k=0}^{\indexK}\abs{\theta_1\indexk{k}-\theta_2\indexk{k}}+ \numK\left(x-x_\indexK\right)\abs{\left(\theta_1\indexk{\indexK+1}-\theta_2\indexk{\indexK+1}\right)}\right)\mathrm{d}x                    \\
		                                   & \leq \sum_{\indexK=0}^{\numK-1}\sum_{k=0}^{\indexK}\abs{\theta_1\indexk{k}-\theta_2\indexk{k}}/\numK + \numK\sum_{\indexK=0}^{\numK-1}\abs{\left(\theta_1\indexk{\indexK+1}-\theta_2\indexk{\indexK+1}\right)}\int_{x_\indexK}^{x_{\indexK+1}}\left(x-x_\indexK\right)\mathrm{d}x \\
		                                   & =  \sum_{\indexK=0}^{\numK-1}\sum_{k=0}^{\indexK}\abs{\theta_1\indexk{k}-\theta_2\indexk{k}}/\numK + \sum_{\indexK=0}^{\numK-1}\abs{\left(\theta_1\indexk{\indexK+1}-\theta_2\indexk{\indexK+1}\right)}/\numK                                                                     \\
		                                   & = \sum_{\indexK=0}^{\numK-1}\sum_{k=0}^{\indexK+1}\abs{\theta_1\indexk{k}-\theta_2\indexk{k}}/\numK                                                                                                                                                                               \\
		                                   & = \sum_{\indexK=0}^{\numK}\frac{\numK+1-\indexK}{\numK}\abs{\theta_1\indexk{\indexK}-\theta_2\indexk{\indexK}}                                                                                                                                                                    \\
		                                   & \leq 2 \|\theta_1-\theta_2 \|_1.
	\end{align*}

	Combining the bounds above, for fixed $n$ and $L$, we get
	\begin{align*}
		\frac{1}{n^2}\sum_{i,j}\mathcal{W}_1\left(w_{ij}, \hat{w}_{ij}\right) & = \frac{1}{n^2}\sum_{i,j}\mathcal{W}_1\left(w_{ij}, \tilde{w}_{ij}\right) + \frac{1}{n^2}\sum_{i,j}\mathcal{W}_1\left(\tilde{w}_{ij}, \hat{w}_{ij}\right) \\
		                                                                      & \leq 2\numK^{-1} +  \frac{2}{n^2}\sum_{i,j}\|\tilde{\theta}_{ij}-\hat{\theta}_{ij}\|_1                                                                    \\
		                                                                      & = 2\numK^{-1} +  \frac{2}{n^2}\|\tilde{\theta}-\hat{\theta}\|_1.
	\end{align*}
\end{proof}

\begin{proof}[Proof of\ \ Theorem~\ref{theorem:rate-compact}]
    Combining Lemma~\ref{lemma:w1_dist} and Theorem~\ref{theorem:holder-rate}, we get that for all $C'>0$, there exists a constant $C>0$ only depending on  $C',M,\alpha$, such that with probability at least $1-\exp\left(-C'n\log(\max(k,s))\right)$, we have
    \begin{align*}
        \frac{1}{n^2}\sum_{i,j}\mathcal{W}_1\left(w_{ij}, \hat{w}_{ij}\right) & \leq 2\numK^{-1} +  \frac{2}{n^2}\|\probshat - \probstilde\|_1 \\
        & \leq  2L^{-1} + 2L^{1/2}\left(\frac{1}{n^2}\|\probshat - \probstilde\|_2^2\right)^{1/2}\\
        & \leq  C\left[L^{-1} + L^{1/2}\left(k^{-2(\alpha \wedge 1)} + Lk^2n^{-2} + \log(n)n^{-1}  \right)^{1/2}\right] \\
        & \leq  C\left[L^{-1} +  L^{1/2}\left(k^{-(\alpha \wedge 1)} + L^{1/2}kn^{-1} + \sqrt{\log(n)}n^{-1/2}  \right)\right].
    \end{align*}
Now let $\gamma,\delta \in (0,1)$,  $k=n^{\delta}$, and $L=n^{\gamma}$. We then have
\begin{align*}
    \frac{1}{n^2}\sum_{i,j}\mathcal{W}_1\left(w_{ij}, \hat{w}_{ij}\right) & \leq C\left[n^{-\gamma} +  n^{\gamma/2}\left(n^{-\delta(\alpha \wedge 1)} + n^{\gamma/2+\delta-1} + \sqrt{\log(n)}\,n^{-1/2}  \right)\right]\\
    & \leq C\left(n^{-\gamma} + n^{-\delta(\alpha \wedge 1)+\gamma/2} + n^{\gamma+\delta-1} + \sqrt{n^{\gamma-1}\log(n)}\right).
\end{align*}
Balancing the first two terms gives $\gamma = 2\delta(\alpha\wedge 1)/3$; the third term is then dominated as soon as $\delta+2\gamma\leq 1$, while the clustering term $\sqrt{n^{\gamma-1}\log(n)}$ satisfies $\sqrt{n^{\gamma-1}\log(n)} = o(n^{-\gamma})$ whenever $\gamma<1/3$ strictly, and is therefore absorbed in the $n^{-\gamma}$ rate by a constant. Letting $\delta = \frac{3}{4\alpha+3}$ and $\gamma = \frac{2\alpha}{4\alpha+3}$, we have $\gamma\leq 2/7 < 1/3$ for $\alpha\in(0,1]$ and all side conditions are satisfied, so for all $C'>0$, there exists a constant $C>0$ only depending on $C',M,$ and $\alpha$, such that
\begin{equation*}
    \frac{1}{n^2}\sum_{i,j}\mathcal{W}_1\left(w_{ij}, \hat{w}_{ij}\right)  \leq Cn^{-2\alpha/(4\alpha+3)},
\end{equation*}
with probability at least $1-\exp\left(-C'n\right)$.
\end{proof}

\subsection{$L_2$ norm between densities}

We now suppose that the probability measures we consider have smooth densities, with smoothness independent of the latent variables. For distribution functions $F,F'$ with densities $f,f'$, we consider  $\operatorname{d}(F,F') = \|f-f'\|_2^2$, where $\|f\|_2$ is the $L^2$ norm of $f$ with respect to the Lebesgue measure.

\begin{assumption}
	\label{assumption:density_function}
	$W:[0,1]^2 \rightarrow \measures{[0,1]}$ is such that, for all $\xi_1,\xi_2\in [0,1]$, the distribution function $F$ associated to $W(\xi_1,\xi_2)$ in Assumption~\ref{assumption:DistributionFunction} is absolutely continuous, with density function $f$ such that $|f(a)-f(b)|\leq B |a-b|$ for some $B$  independent of $\xi_1,\xi_2$.
	Thus $w$ has density $$\mu(x) = (1-p)\delta_0(x) + pf(x).$$
\end{assumption}

Assumption~\ref{assumption:density_function} imposes smoothness assumptions on the image of the decorated graphon, and additionally requires the image to be "globally" smooth since all the distributions $W(u,v)$ will have a continuous density that respects the Lipschitz condition with $B$. This prevents the edge distributions from varying wildly from each other.

\begin{lemma}
	\label{lemma:approx_l2_density}
	Let $w$ be as in Assumption~\ref{assumption:density_function}, with density $\mu$. Let $\tilde{w},\hat{w} \in \histmeasures[L]{[0,1]}$ with $\tilde{w}$ the discretization of $w$ as defined in eq.~\eqref{eq:theta_hist_from_cont} and $\tilde{\mu},\hat{\mu}$ their densities. Then for some constant $C>0$, we have
	\begin{equation*}
		\|\mu-\hat{\mu}\|_2^2 \leq C\left(L^{-2} + L\|\tilde{\theta}-\hat{\theta}\|_2^2\right),
	\end{equation*}
	where $C=\max\{3B^2,1\}$ with $B$ from Assumption~\ref{assumption:density_function}.
\end{lemma}

\begin{proof}[Proof of\ \ Lemma~\ref{lemma:approx_l2_density}]
	We measure the distance between $w$ and $\hat{w}$ by the $L^2$ norm between their densities $\mu,\hat{\mu}$

	$$d(w,\hat{w}) = \|\mu-\hat{\mu}\|_2^2 = \int_0^1 (\mu(x)-\hat{\mu}(x))^2dx \leq 3\left(\|\mu-\tilde{\mu}\|_2^2 + \|\hat{\mu}-\tilde{\mu}\|_2^2\right).$$

	We start with the discretization error $\|\mu-\tilde{\mu}\|_2^2$: the atom at $0$ is defined to be the same in $w$ and $\tilde{w}$ (see eq.~\eqref{eq:def_w_tilde}), hence we only have to focus on the continuous term
	\begin{align*}
		\|\mu-\tilde{\mu}\|_2^2 & = \int_{0}^{1}\left(\mu(x)-\tilde{\mu}(x)\right)^2dx                                                  \\
		                        & \leq \int_{0}^{1}\left(f(x)- \tilde{f}(x)\right)^2dx                                                  \\
		                        & \leq \sum_{l=1}^L\int_{x_{l-1}}^{x_l}\left(f(x)- L\int_{x_{l-1}}^{x_l}{f}(u)du\right)^2dx             \\
		                        & \leq \sum_{l=1}^L\int_{x_{l-1}}^{x_l}\left(L\int_{x_{l-1}}^{x_l}\left|f(x)-{f}(u)\right|du\right)^2dx \\
		                        & \leq \sum_{l=1}^L\int_{x_{l-1}}^{x_l}\left(L\int_{x_{l-1}}^{x_l}B\left|x-u\right|du\right)^2dx        \\
		                        & \leq B^2L^{-2}
	\end{align*}

	We now relate the norm of $\tilde{w},\hat{w} \in \histmeasures{[0,1]}$ to the distance between their parameters $\tilde{\theta},\hat{\theta}$ by noticing that these functions are piecewise constant and using the definition of the integral of the delta function
	\begin{align*}
		\|\tilde{\mu}-\hat{\mu}\|_2^2 = \int_{0}^{1}\left(\tilde{\mu}(x)-\hat{\mu}(x)\right)^2dx
		=\left(\tilde{\theta}\indexk{0}-\hat{\theta}\indexk{0}\right)^2 + \sum_{l=1}^{L}L\left(\tilde{\theta}\indexk{l}-\hat{\theta}\indexk{l}\right)^2
		\leq L\|\tilde{\theta}-\hat{\theta}\|_2^2.
	\end{align*}
\end{proof}

\begin{remark}
	Modifying the Assumption~\ref{assumption:density_function} by requiring the densities to be Hölder continuous with exponent at least $0<\eta\leq 1$ (i.e., $|f(a)-f(b)|\leq B |a-b|^\eta$) changes Lemma~\ref{lemma:approx_l2_density} in a straightforward way:
	\begin{equation*}
		\|\mu-\hat{\mu}\|_2^2 \leq C\left(L^{-2\eta} + L\|\tilde{\theta}-\hat{\theta}\|_2^2\right).
	\end{equation*}
\end{remark}

Combining the H\"older-rate bounds from Appendix~\ref{subsection:holder-rate} with the discretization error yields the following.

\begin{theorem}
	\label{theorem:convergence_l2_densities}
	For $W \in \mathcal{H}(\alpha, M)$ respecting Assumption~\ref{assumption:density_function}, for any $C^{\prime}>0$, there exists a constant $C>0$ only depending on $C^{\prime}, M,B, \alpha$, where the following holds
	$$
		\frac{1}{n^2} \sum_{i j}\|\mu_{ij}-\hat{\mu}_{i j}\|_2^2 \leq Cn^{-4\alpha/(4\alpha+3)},
	$$
	with probability at least $1-\exp \left(-C^{\prime} n\right)$, with $k = \left\lceil n^{3/(4\alpha+3)} \right\rceil$, $s=\binom{k}{2}$ and $L=\left\lceil n^{2\alpha/(4\alpha+3)}\right\rceil$.
\end{theorem}

\begin{proof}[Proof of\ \ Theorem~\ref{theorem:convergence_l2_densities}]
	From Lemma~\ref{lemma:approx_l2_density}, we have
	\begin{equation*}
		\frac{1}{n^2} \sum_{i j}\|\mu_{ij}-\hat{\mu}_{i j}\|_2^2 \leq 3B^2L^{-2} + \frac{L}{n^2}\sum_{ij}\|\tilde{\theta}_{ij}-\hat{\theta}_{ij}\|_2^2.
	\end{equation*}
	The second term on the right-hand side can then be analyzed using the results from Section~\ref{section:inference_finitely}. We reuse the proof of Theorem~\ref{theorem:holder-rate}, and we obtain that there exists a constant $C_2>0$
	\begin{equation*}
		\frac{1}{n^2} \sum_{i j}\|\mu_{ij}-\hat{\mu}_{i j}\|_2^2 \leq 3B^2L^{-2} + C_2\left(L\left(\frac{k}{d}\right)^{-2(\alpha \wedge 1)} + L^2s^*n^{-2} + L\log(\max(k,s^*))n^{-1}\right),
	\end{equation*}
	with probability at least $1-\exp(-C'n)$. We now first focus on the block model approximation: we let the number of shapes $s$ and the number of bins $L$ depend on $n$: $k=n^\delta, s\propto n^{2\delta}$ and $L=n^\gamma$ for $\delta,\gamma \in (0,1]$ and $\beta=1$. We obtain that for a $C_3>0$,
	\begin{equation*}
		\frac{1}{n^2} \sum_{i j}\|\mu_{ij}-\hat{\mu}_{i j}\|_2^2 \leq C_3\left(n^{-2\gamma} + n^{-2\delta(\alpha \wedge 1)+\gamma} + n^{-2+2\gamma+2\delta} + n^{-1+\gamma} \log(n)\right).
	\end{equation*}
	To find the best rate, we equate the first and second terms on the right-hand side. This leads to $\gamma = 2\delta(\alpha\wedge 1)/3$, and setting $\delta=3/(4(\alpha\wedge 1)+3)$ balances the first three terms. Since the resulting $\gamma = 2(\alpha\wedge 1)/(4(\alpha\wedge 1)+3) \leq 2/7$ satisfies $\gamma<1/3$ strictly, the clustering term $n^{-1+\gamma}\log(n) = o(n^{-2\gamma})$, and is therefore absorbed in the $n^{-2\gamma}$ rate by a constant. For $\alpha\in(0,1]$ this yields the required rate.

	For the bound based on stochastic shape models, we consider elements of $\boundedssm$, and we let the number of blocks $k$, their diameter $d$ and the number of bins $L$ depend on $n$: $k=n^{\delta}, d=n^{\beta}$ and $L=n^\gamma$ for $\delta,\gamma \in (0,1]$. We obtain that for a $C_3>0$,
	\begin{equation*}
		\frac{1}{n^2} \sum_{i j}\|\mu_{ij}-\hat{\mu}_{i j}\|_2^2 \leq C_3\left(n^{-2\gamma} + n^{-2(\delta-\beta)(\alpha \wedge 1)+\gamma} + n^{-2+2\gamma+2(\delta-\beta)} + n^{-1+\gamma} \log(n)\right).
	\end{equation*}
	Thus for any $\delta,\beta$ such that $\delta-\beta = 3/(4\alpha+3)$, we obtain the same rates as with the block model approximation. The number of shapes will be of the order of $n^{6/(4\alpha+3)}$ with diameter proportional to $(n^\beta+2)/n^\delta \approx n^{\beta-\delta}$.
\end{proof}

Using a direct adaptation of Klopp et al. \citep[Proposition 3.5]{klopp_oracle_2017}, we obtain that the agnostic error is $n^{-\alpha \wedge 1}$, and thus

\begin{lemma}
	Under the same setting as in Theorem~\ref{theorem:convergence_l2_densities}, we get that
	$$
		\operatorname{MISE}\left(\hat{W}_{\probshat}, W\right) \leq C\left(n^{-4 \alpha /(4\alpha+3)} + n^{-\alpha\wedge 1}\right).
	$$
	The first term dominates for $\alpha\geq 1/4$, recovering the rate of Theorem~\ref{theorem:convergence_l2_densities}; for $\alpha\in(0,1/4)$ the agnostic term $n^{-\alpha}$ is the slower of the two.
\end{lemma}

\section{Optimization: Practical Considerations}
\label{sec:optimization}

\begin{remark}[Computational efficiency]
    \label{remark:computation}
	The choice of indicator functions in eq.~\eqref{eq:estimator_w_reconstruction} allows us to approximate the solution of the least squares problem ``efficiently'' via greedy label-switching algorithms \citep{zhao_consistency_2012,glover_tabu_1997} as done by Olhede and Wolfe \citep{olhede_network_2014}. Starting from a good partition of the nodes, testing a local update (switching the labels of two nodes) can be done in $O(Lkn)$ time in the worst case. Other methods such as, Kernel density estimation (KDE) could be considered; however exact computation of the KDE is expensive, and fast approximations use binning strategies. Additionally, testing for a local update with KDE would require refitting the estimated density for each pair of blocks and then evaluating the log-likelihood, leading to a complexity of at best $O(k^2n\log(n))$ using a fast Fourier transform.
\end{remark}

The results presented in Section~\ref{subsection:finitely-properties} are valid for a global optimum of eq.~\eqref{eq:least-squares-formulation}. This problem boils down to finding an optimal partition of the vertices, which is, in principle, NP-hard. We approximate this global optimum using a greedy label-switching algorithm \citep{bickel_nonparametric_2009,olhede_network_2014}. The starting point of this label-switching algorithm can influence whether we end up in a local optimum, but picking a good starting point instead of a random one helps to reduce the optimality gap in practice. Standard starting points are often computed using spectral clustering \citep{olhede_network_2014,arroyo_inference_2021}, but we found that experimentally using a multilevel $k$-way partitioning \citep{karypis_multilevel_1998} as implemented in the library METIS \citep{karypis_metis_1997} is fast and provides very good starting points.

\begin{remark}
	Barbillon et al. \citep{barbillon_stochastic_2017} use variational inference \citep{celisse_consistency_2012} to fit the multiplex stochastic block model, while our code \citep{dufour_networkhistogramjl_2023} uses least squares. Gaucher and Klopp \citep{gaucher_maximum_2021,gaucher_optimality_2021} discuss the relationship between the two estimators and show that both are minimax optimal in the context of graphon estimation.
\end{remark}

To pick the resolution $k$ when $\alpha$ is unknown, we use the automatic bandwidth estimation of Olhede and Wolfe \citep{olhede_network_2014} on each decoration probability $w\indexk{l}$, yielding $\hat{k}_l$, and pick the biggest group number $\hat{k} = \max_l \hat{k}_l$. This is equivalent to taking the least regular $w\indexk{l}$ to decide the number of groups needed. The number of shapes $s$ is picked using the Bayesian Information Criterion (BIC) as in Verdeyme and Olhede \citep{verdeyme_hybrid_2024}.

\section*{Supplementary Material}
\textbf{Code and data.} Code to reproduce the experiments of the paper. The methods are implemented in the Julia package \texttt{NetworkHistogram.jl} \citep{dufour_networkhistogramjl_2023}.

\end{document}